\documentclass[onecolumn]{IEEEtran}
\IEEEoverridecommandlockouts
\usepackage{multicol}
\usepackage{amsmath}
\usepackage{mathtools}
\usepackage{amsfonts}
\usepackage{pifont}
\usepackage{amssymb}
\usepackage{amsthm}
\usepackage{tikz}
\usepackage{cases}
\usepackage{graphicx}
\usepackage{float,stfloats}
\usepackage{subcaption}
\usepackage{multirow}
\graphicspath{{../}}
\DeclareGraphicsExtensions{.pdf}
\usepackage[colorlinks,
            linkcolor=blue,
            anchorcolor=blue,
            citecolor=blue]{hyperref}
\usepackage{multirow}
\usepackage{booktabs}
\usepackage{url}
\usepackage{xtab}
\usepackage{tabu}
\usepackage{longtable}
\usepackage{algorithm}
\usepackage{algorithmic}
\usepackage{enumerate}
\usepackage{makecell}
\usepackage{lipsum}
\usepackage{multicol}
\usepackage{mathdots}
\usepackage{extarrows}
\usepackage{color,xcolor}
\usepackage{bm}
\usepackage{bbding}
\usepackage{pifont}
\usepackage{arydshln}
\usetikzlibrary{arrows}
\usepackage[numbers]{natbib}
\theoremstyle{plain}

\newtheorem{theorem}{Theorem}
\newtheorem{lemma}[theorem]{Lemma}
\newtheorem{proposition}[theorem]{Proposition}
\newtheorem{definition}[theorem]{Definition}
\newtheorem{corollary}[theorem]{Corollary}
\theoremstyle{definition}
\newtheorem{example}{Example}
\newtheorem{construction}{Construction}
\newtheorem{remark}{Remark}

\def\BibTeX{{\rm B\kern-.05em{\sc i\kern-.025em b}\kern-.08em
		T\kern-.1667em\lower.7ex\hbox{E}\kern-.125emX}}

\begin{document}

\begin{sloppypar}
\title{Calculating the I/O Cost of Linear  Repair Schemes for RS Codes Evaluated on Subspaces via Exponential Sums}
	
\author{Zhongyan Liu, Jingke Xu, Zhifang Zhang%
	\thanks{Zhongyan Liu and Zhifang Zhang are both with the Academy of Mathematics and Systems Science, CAS, and University of Chinese Academy of Sciences, Beijing, China. e-mail: liuzhongyan@amss.ac.cn, zfz@amss.ac.cn.%
		
		Jingke Xu is with School of Information Science and Engineering, Shandong Agricultural University, Tai’an 271018, China. e-mail: xujingke@sdau.edu.cn.}}
\maketitle
	
\thispagestyle{empty}
\begin{abstract}
The I/O cost, defined as the amount of data accessed at helper nodes during the repair process, is a crucial metric for repair efficiency of Reed-Solomon (RS) codes. Recently, a formula that relates the I/O cost to the Hamming weight of some linear spaces was proposed in [Liu\&Zhang-TCOM2024]. In this work, we introduce an effective method for calculating the Hamming weight of such linear spaces using exponential sums. With this method, we derive lower bounds on the I/O cost for RS codes evaluated on a $d$-dimensional subspace of $\mathbb{F}_{q^\ell}$ with $r=2$ or $3$ parities. These bounds are exactly matched in the cases $r=2,\ell-d+1\mid\ell$ and $r=3,d=\ell$ or $\ell-d+2\mid\ell$, via the repair schemes designed in this work. We refer to schemes that achieve the lower bound as I/O-optimal repair schemes. Additionally, we characterize the optimal repair bandwidth of I/O-optimal repair schemes for full-length RS codes with two parities, and build an I/O-optimal repair scheme for full-length RS codes with three parities, achieving lower repair bandwidth than previous schemes.
\end{abstract}
	
\section{Introduction}\label{Sec1}
To ensure fault-tolerant storage with low redundancy, Maximum Distance Separable (MDS) codes are extensively used in distributed storage systems (DSSs). Specifically, a data file of $k$ blocks is encoded into a codeword of $n$ blocks using an $[n,k]$ MDS code, which is then distributed across $n$ storage nodes each storing one block. Due to the frequent occurrence of node failures, node repair is a central issue in code-based DSSs, where data stored on failed nodes must be recovered by downloading information from surviving nodes (i.e., helper nodes). A trivial repair approach for the $[n,k]$ MDS code involves connecting to any $k$ helper nodes and downloading all the data stored on these nodes. However, the trivial scheme incurs a high repair cost, as the repair bandwidth equals the entire file size. In the repair process, the repair bandwidth refers to the total amount of data transmitted from the helper nodes. Minimizing the repair bandwidth becomes one of the driving forces behind research in distributed storage codes. Dimakis et al. \cite{csbound} established the cut-set bound to characterize the optimal repair bandwidth for MDS array codes, which has since prompted extensive research aimed at developing MDS array codes with optimal repair bandwidth \cite{overview}.

Reed-Solomon (RS) codes are the most widely used family of MDS codes to date, and thus constructing efficient repair schemes for RS codes is of great significance to practical use. Guruswami and Wootters \cite{RSrepair} proposed the first repair scheme for RS codes that achieves a lower repair bandwidth than the trivial approach. Their scheme is based on the framework of linear repair schemes for scalar MDS codes proposed by Shanmugam et al. in \cite{scalarMDS}, i.e., treating scalar MDS codes over $\mathbb{F}_{q^\ell}$ as MDS array codes over $\mathbb{F}_q$ with sub-packetization $\ell$. Later, Dau and Milenkovic \cite{obRS} extended the scheme in \cite{RSrepair} to a broader parameter regime. Although the schemes in \cite{RSrepair,obRS} have been proven to achieve the optimal repair bandwidth in the full-length cases, they still fall short of the cut-set bound due to the constraint $\ell={\rm O}(\log n)$.  Tamo et al. \cite{RScsbound1} first designed the RS code with a repair bandwidth that matches the cut-set bound for sufficiently large $\ell\gtrsim n^n$. Then, a series of works \cite{RStradeoff,RStradeoff2,RStradeoff3} are devoted to providing a tradeoff between the sub-packetization and repair bandwidth. Also for practical reasons, some works \cite{RS1410, RSoverF2} studied repair schemes for RS codes with the parameters currently used in modern storage systems. In a recent work \cite{subspacelocators}, Berman et al. generalized the repair schemes in \cite{obRS} to repair RS codes evaluated on a subspace, leading to a lower repair bandwidth than \cite{obRS}. Additionally, the repair of multiple node failures for RS codes has been studied both in the centralized model \cite{CenRS} and the cooperative model \cite{CoopRS}.

Besides repair bandwidth, the I/O cost, which is the volume of data accessed at all helper nodes during the repair process, is also an important metric for repair efficiency. The I/O cost must be at least as high as the repair bandwidth, so a lower bound on the repair bandwidth naturally provides a lower bound on the I/O cost. However, calculating the I/O cost is more complex because it is also influenced by the selection of the basis of $\mathbb{F}_{q^\ell}$ over $\mathbb{F}_q$. Dau et al. \cite{fullr=2} first derived a lower bound on the I/O cost for repairing full-length RS codes with two parities, improving upon the repair bandwidth bound established in \cite{RSrepair,obRS}. They also presented repair schemes that achieve this bound. However, their bound and repair schemes are limited to the field of characteristic 2. Later, Li et al. \cite{shortr=2} extended the lower bound to RS codes evaluated on a subspace with two parities over finite fields of characteristic 2, but failed to provide matching repair schemes. Based on the construction in \cite{RScsbound1}, Chen et al. \cite{oaRS} built a family of RS codes by further enlarging $\ell$ by a factor exponential in $n$, ensuring that both the I/O cost and repair bandwidth meet the cut-set bound. Recently, a formula that relates the I/O cost to the Hamming weight of some linear spaces was proposed in \cite{I/OFormula}. Using this formula, the authors in \cite{I/OFormula} established lower bounds on the I/O cost for full-length RS codes with two or three parities. Moreover, they constructed linear repair schemes for full-length RS codes via $q$-polynomials, achieving a reduced I/O cost compared to the schemes in \cite{RSrepair,obRS}.

Although both the scalar codes and array codes have been constructed such that their I/O cost and repair bandwidth simultaneously achieve the cut-set bound \cite{oaRS,longMSR,Ye-OAMSR}, it is proved that they require a large sub-packetization with $\ell\geq (n-k)^{\lceil\frac{n-1}{n-k}\rceil}$ \cite{spbound}. For small $\ell$, the tradeoff between repair bandwidth and I/O cost is still unknown. A first attempt along this line of research was made in \cite{IOfulllength}, where the authors showed that the bandwidth-optimal repair schemes for the full-length RS codes in \cite{RSrepair,obRS} incur a trivial I/O cost. Later in \cite{fullr=2}, the authors proved that the I/O cost of any bandwidth-optimal repair schemes for full-length RS codes with two parities over finite fields of characteristic 2 is trivial.

\subsection{Contributions}
In this work, we focus on the I/O cost of linear repair schemes for RS codes evaluated on a $d$-dimensional subspace of $\mathbb{F}_{q^\ell}$. Both the I/O cost and repair bandwidth are measured in the number of symbols in $\mathbb{F}_q$ throughout. Our contributions are outlined as follows.
\begin{enumerate}
\item \textbf{An effective method for calculating the I/O cost.} Specifically, we define $\ell$ special parity-check polynomials corresponding to each linear repair scheme, referred to as normalized polynomials (see Definition \ref{def}). These normalized polynomials simplify the calculation of I/O cost, and play an important role in subsequently characterizing the optimal repair bandwidth for I/O-optimal repair schemes. Based on the normalized polynomials, we utilize exponential sums to calculate the I/O cost and obtain a concise formula in Theorem \ref{thmformula}.
\item \textbf{Improved lower bounds on the I/O cost.} By applying our formula and the Weil bound for exponential sums, we first derive a lower bound on the I/O cost for full-length RS codes over $\mathbb{F}_{q^\ell}$ with $r\leq {\rm Char}(\mathbb{F}_{q^\ell})$ parities, improving upon the lower bound of repair bandwidth derived in \cite{RSrepair,obRS}. Furthermore, we provide a more precise estimate on the I/O cost for RS codes evaluated on a $d$-dimensional subspace of $\mathbb{F}_{q^\ell}$ with two or three parities, establishing lower bounds in Theorem \ref{r=2} and Theorem \ref{r=3}, respectively. In the cases where $r=2,\ell-d+1\mid\ell$ and $r=3$, $d=\ell$ or $\ell-d+2\mid\ell$, our bounds are tight due to the repair schemes built in Construction \ref{cons2}. The comparison between our bounds and all known bounds on the I/O cost is presented in Table \ref{table1}.

\item \textbf{Lower bounds on the repair bandwidth for I/O-optimal repair schemes of RS codes with two or three parities.} We fully determine the optimal repair bandwidth of I/O-optimal repair schemes for full-length RS codes with two parities. For full-length RS codes with three parities, we build a repair scheme in Construction \ref{cons1}, which has a lower repair bandwidth than the scheme in \cite{I/OFormula} while achieving the optimal I/O cost. The results on the repair bandwidth derived in this work are summarized in Table \ref{table-RB}.
\end{enumerate}

\begin{table}[H]
	\renewcommand\arraystretch{1.45}
	\centering
	\captionsetup{justification=centering}
	\caption{\scriptsize Lower bounds on the I/O cost for RS code ${\rm RS}(\mathcal{A},n-r)$ over $\mathbb{F}_{q^\ell}$,\\ where $\mathcal{A}$ is a $d$-dimensional $\mathbb{F}_{q}$-subspace of $\mathbb{F}_{q^\ell}$ and $n=q^d$. }\label{table1}
	\begin{tabular}{|c|c|c|c|c|c|}
		\hline  $r$ & Reference & $d$ & Lower bound of I/O & Tight: Y/N & $q$\\ \hline
		\multirow{2}{*}{$r\leq {\rm Char}(\mathbb{F}_{q^\ell})$}  & \cite{RSrepair,obRS} & \multirow{2}{*}{$d=\ell$} & $(n-1)\ell-\frac{r-1}{q-1}(q^\ell-1)$ & N &\multirow{2}{*}{ALL}  \\ \cline{2-2}\cline{4-5}
		& {\bf Cor. \ref{coro11}} &   & $(n\!-\!1)\ell-q^{\ell\!-\!1}\!-(r\!-\!2)(q\!-\!1)q^{\frac{\ell}{2}-1}$    &  $r=2$ & \\ \hline
		\multirow{3}{*}{$r=2$}     &\citep[Thm. 6]{I/OFormula}  & $d=\ell$ & $(n-1)\ell-q^{\ell-1}$ & Y & ALL\\ \cline{2-6}
		&\citep[Thm. 1]{shortr=2}  & \multirow{2}{*}{$d\leq\ell$} & $(n-1)\ell-(\ell-d+1)2^{d-1}$ &  ${\ell\!-\!d\!+\!1\mid\ell}^{~*}$ &  2\\ \cline{2-2}\cline{4-6}
		& {\bf Thm. \ref{r=2}}   &  & $ (n-1)\ell-(\ell-d+1)q^{d-1}$ &  $\ell\!-\!d\!+\!1\mid\ell$ & ALL\\ \hline
		\multirow{2}{*}{$r=3$}     & \citep[Thm. 7]{I/OFormula}  & $d=\ell$ &$(n-1)\ell-2^\ell-2^{\ell-3}$ & N & \multirow{2}{*}{2}\\  \cline{2-5}
		& {\bf Thm. \ref{r=3}} & $d\leq\ell$ & $(n-1)\ell-(\ell-d+2)2^{d-1}$ & $d=\ell$ or $\ell\!-\!d+2\mid\ell$ & \\ \hline
		
	\end{tabular}
	
	\vspace{6pt}
		\footnotesize{$^*$\cite{shortr=2} derived the same lower bound as our Thm. \ref{r=2}, however, it restricted to ${\rm Char}(\mathbb{F}_{q^\ell})=2$ and did not provide a repair scheme matching the lower bound. In the column Tight: Y/N,  Y means the bound is tight for all parameters, N means the bound is not tight,  and the parameter conditions indicate when the bound is tight.}
\end{table}

\begin{table}[H]
	\renewcommand\arraystretch{1.45}
	\centering
	\captionsetup{justification=centering}
	\caption{\scriptsize Lower bounds on the repair bandwidth of I/O-optimal repair schemes for ${\rm RS}(\mathcal{A},n-r)$ over $\mathbb{F}_{q^\ell}$, \\ where $\mathcal{A}$ is a $d$-dimensional $\mathbb{F}_{q}$-subspace of $\mathbb{F}_{q^\ell}$ and $n=q^d$. }\label{table-RB}
	\begin{tabular}{|c|c|c|c|c|c|}
		\hline
		\multicolumn{1}{|c|}{$r$} &  This paper &$d$    & $q$  & Lower bound of repair bandwidth & Tight: Y/N \\ \hline
		\multirow{3}{*}{$r=2$} &\multirow{3}{*}{\bf Thm. \ref{b-r=2-d}}  & \multirow{2}{*}{$d=\ell$}    &$>2$    &   $(n-1)\ell-q^{\ell-1}$       & Y    \\ \cline{4-6}
		
		& & &2  &   $(n-1)\ell-3\cdot 2^{\ell-2}$       & Y     \\ \cline{3-6}
		& & $d<\ell,\ell-d+1|\ell$ &      ALL & $(n-1)d-q^{2d-\ell-1}$ &Unknown \\ \hline
		\multicolumn{1}{|c|}{$r=3$} & {\bf Thm. \ref{b-r=3-d}} & $d=\ell$ or $\ell\!-\!d\!+\!2\mid\ell$   & 2  & $(n-1)(d-1)-2^{2d-\ell-1}+\lfloor2^{3d-2\ell-4}\rfloor$  & Unknown \\ \hline
	\end{tabular}
\end{table}

\subsection{Organization}	
The remaining of the paper is organized as follows. Section \ref{Sec2} introduces preliminaries of repair schemes. Section \ref{Sec3} introduces exponential sums for an effective estimation of the I/O cost, and thereby derives Theorem \ref{thmformula} and Corollary \ref{coro11}. Section \ref{Sec4} further derives lower bounds on the I/O cost for RS codes with two or three parities, and also characterizes the repair bandwidth for I/O-optimal schemes. Section \ref{Sec5} presents the construction of I/O-optimal repair schemes. Finally, Section \ref{Sec6} concludes the paper.

\section{Preliminaries}\label{Sec2}
For positive integers $m\leq n$, denote $[n]=\{1,...,n\}$ and $[m,n]=\{m,m+1,...,n\}$. Let $B=\mathbb{F}_q$ be the finite field of $q$ elements and $F=\mathbb{F}_{q^\ell}$ be the extension field of $B$ with degree $\ell>1$. For any $\alpha\in F$,
the trace ${\rm Tr}_{F/B}(\alpha)$ of $\alpha$ over $B$ is defined by ${\rm Tr}_{F/B}(\alpha)=\sum^{\ell-1}_{i=0}\alpha^{q^i}$. If $B$ is the prime subfield of $F$, then ${\rm Tr}_{F/B}(\alpha)$ is called the absolute trace of $\alpha$ and simply denoted by ${\rm Tr}(\alpha)$. All vectors throughout are treated as row vectors and denoted by bold letters in italics, such as ${\bm c},{\bm g}$, etc.  For a vector of length $m$, say, $\bm{x}=(x_1,...,x_m)$, define ${\rm supp}(\bm{x})=\{j\in[m]:x_j\neq 0\}$ and ${\rm wt}(\bm{x})=|{\rm supp}(\bm{x})|$. Furthermore, for a set of vectors $W\subseteq B^m$, define ${\rm supp}(W)=\bigcup_{{\bm x}\in W}{\rm supp}({\bm x})$ and ${\rm wt}(W)=\sum_{\bm{x}\in W}{\rm wt}(\bm{x})$. Moreover, let ${\rm span}_{B}(W)$ be the $B$-linear space spanned by the vectors in $W$, and $\dim_B(W)$ denote the rank of a set of vectors in $W$ over $B$.

\subsection{Linear repair schemes for scalar MDS codes}

First recall some basics about the vector representation of elements of $F$ over $B$.
Let $\mathcal{B}=\{\beta^{(1)},...,\beta^{(\ell)}\}$ be a basis of $F$ over $B$ and $\hat{\mathcal{B}}=\{\gamma^{(1)},...,\gamma^{(\ell)}\}$ be the dual basis of $\mathcal{B}$. It is well known that $\alpha=\sum_{i=1}^\ell{\rm Tr}_{F/B}(\alpha\gamma^{(i)})\beta^{(i)}$ for any $\alpha\in F$. Thus, the vector representation of elements of $F$ over $B$ with respect to $\mathcal{B}$ is defined by a map $\Phi_{\mathcal{B}}: F\rightarrow B^\ell$ where \begin{equation*}\Phi_{\mathcal{B}}(\alpha)=({\rm Tr}_{F/B}(\alpha\gamma^{(1)}),...,{\rm Tr}_{F/B}(\alpha\gamma^{(\ell)}))\;,~\forall \alpha\in F.\label{eq1-}\end{equation*}
For simplicity, the elements in $F$ are called symbols, and those in $B$ are called subsymbols.

Let $\mathcal{C}$ be an $[n,k]$ linear scalar MDS code over $F$. To repair a failed  node $i^*$ which stores $\textbf{c}_{i^*}$ for each codeword $(\textbf{c}_1,...,\textbf{c}_n)\in\mathcal{C}$, the trivial repair scheme is to download any $k$ symbols in $\{\textbf{c}_{i}\}_{i\neq i^*}$. A nontrivial repair scheme is to vectorize the scalar MDS code and download partial subsymbols from each helper node. More specifically,
\begin{equation*}
   \Phi_{\mathcal{B}}(\mathcal{C})=\{(\Phi_{\mathcal{B}}(\textbf{c}_1),...,\Phi_{\mathcal{B}}(\textbf{c}_n)):(\textbf{c}_1,...,\textbf{c}_n)\in\mathcal{C}\}	
\end{equation*}
is an $(n,k;\ell)$ linear MDS array code over $B$. Then the node repair problem falls into the linear array code $\Phi_{\mathcal{B}}(\mathcal{C})$.
\begin{lemma}\label{lem1}
	The dual code of $\Phi_{\mathcal{B}}(\mathcal{C})$ is $\Phi_{\hat{\mathcal{B}}}(\mathcal{C}^\bot)$.
\end{lemma}
\begin{proof}
We first prove that $\Phi_{\hat{\mathcal{B}}}(\theta)\Phi_{\mathcal{B}}(\alpha)^\top={\rm Tr}_{F/B}(\theta\alpha)$ for all $\theta,\alpha\in F$. Writing $\theta$ and $\alpha$ as the combinations with respect to the basis $\hat{\mathcal{B}}$ and $\mathcal{B}$, respectively, i.e.,  $\theta=\Phi_{\hat{\mathcal{B}}}(\theta)(\gamma^{(1)},...,\gamma^{(\ell)})^\top$ and $\alpha=(\beta^{(1)},...,\beta^{(\ell)})\Phi_{\mathcal{B}}(\alpha)^\top$, it follows $$\theta\alpha=\Phi_{\hat{\mathcal{B}}}(\theta)\Big((\gamma^{(1)},...,\gamma^{(\ell)})^\top\cdot(\beta^{(1)},...,\beta^{(\ell)})\Big)\Phi_{\mathcal{B}}(\alpha)^\top\;.$$
Due to the $B$-linearity of ${\rm Tr}_{F/B}$, it has ${\rm Tr}_{F/B}(\theta\alpha)=\Phi_{\hat{\mathcal{B}}}(\theta)\cdot\big({\rm Tr}(\gamma^{(i)}\beta^{(j)})\big)_{i,j}\cdot\Phi_{\mathcal{B}}(\alpha)^\top=\Phi_{\hat{\mathcal{B}}}(\theta)\Phi_{\mathcal{B}}(\alpha)^\top$
where the last equality is because $\big({\rm Tr}(\gamma^{(i)}\beta^{(j)})\big)_{i,j}$ equals the identity matrix.

Then, for every ${\bm c}=(\textbf{c}_1,...,\textbf{c}_n)\in\mathcal{C}$ and ${\bm g}=(\textbf{g}_1,...,\textbf{g}_n)\in\mathcal{C}^\perp$,
$$\Phi_{\hat{\mathcal{B}}}({\bm g})\Phi_{\mathcal{B}}({\bm c})^\top=\sum_{j=1}^{n}\Phi_{\hat{\mathcal{B}}}(\textbf{g}_j)\Phi_{\mathcal{B}}(\textbf{c}_j)^\top=\sum_{j=1}^{n}{\rm Tr}_{F/B}(\textbf{g}_j\textbf{c}_j)={\rm Tr}_{F/B}\big(\sum_{j=1}^{n}\textbf{g}_j\textbf{c}_j\big)=0.$$
Therefore, $\Phi_{\hat{\mathcal{B}}}(\mathcal{C}^\perp)\subseteq\Phi_{\mathcal{B}}(\mathcal{C})^\perp$. Noting that $\dim_B(\Phi_{\hat{\mathcal{B}}}(\mathcal{C}^\perp))=(n-k)\ell=\dim_B(\Phi_{\mathcal{B}}(\mathcal{C})^\perp)$, the proof completes.
\end{proof}

Next we introduce the linear repair scheme for node $i^*$.
Suppose ${\bm c}=(\textbf{c}_1,...,\textbf{c}_n)\in\mathcal{C}$ and ${\bm g}^{(j)}=(\textbf{g}_1^{(j)},\textbf{g}_2^{(j)},...,\textbf{g}_n^{(j)})\in\mathcal{C}^\bot$ for all $j\in[\ell]$. By Lemma \ref{lem1}, it has that for $j\in[\ell]$,
$$\Phi_{\hat{\mathcal{B}}}({\bm g}^{(j)})\Phi_{\mathcal{B}}({\bm c})^\top=\sum_{i=1}^{n}\Phi_{\hat{\mathcal{B}}}(\textbf{g}_i^{(j)})\Phi_{\mathcal{B}}(\textbf{c}_i)^\top=0.$$That is,
\begin{equation}\label{eq1}
	W_{i^*}\Phi_{\mathcal{B}}(\textbf{c}_{i^*})^\top=-\sum_{i\neq i^*}W_i\Phi_{\mathcal{B}}(\textbf{c}_{i})^\top\;,
\end{equation}
where for $i\in[n]$,
\begin{equation}\label{W_j}
	W_i=\begin{pmatrix}
		\Phi_{\hat{\mathcal{B}}}(\textbf{g}_{i}^{(1)})\\
		\Phi_{\hat{\mathcal{B}}}(\textbf{g}_{i}^{(2)})\\
		\vdots\\
		\Phi_{\hat{\mathcal{B}}}(\textbf{g}_{i}^{(\ell)})
	\end{pmatrix}=\begin{pmatrix}
	{\rm Tr}_{F/B}(\textbf{g}_{i}^{(1)}\beta^{(1)}) & {\rm Tr}_{F/B}(\textbf{g}_{i}^{(1)}\beta^{(2)}) & \cdots &{\rm Tr}_{F/B}(\textbf{g}_{i}^{(1)}\beta^{(\ell)})\\
	{\rm Tr}_{F/B}(\textbf{g}_{i}^{(2)}\beta^{(1)})& {\rm Tr}_{F/B}(\textbf{g}_{i}^{(2)}\beta^{(2)})& \cdots &{\rm Tr}_{F/B}(\textbf{g}_{i}^{(2)}\beta^{(\ell)})\\
	\vdots & \vdots & \ddots & \vdots \\
	{\rm Tr}_{F/B}(\textbf{g}_{i}^{(\ell)}\beta^{(1)})&{\rm Tr}_{F/B}(\textbf{g}_{i}^{(\ell)}\beta^{(2)})&\cdots &{\rm Tr}_{F/B}(\textbf{g}_{i}^{(\ell)}\beta^{(\ell)})
	\end{pmatrix}.
\end{equation}
We say $\{{\bm g}^{(j)}\}_{j=1}^\ell\subseteq \mathcal{C}^\bot$ defines a repair scheme for node $i^*$, if one can solve $\Phi_{\mathcal{B}}(\textbf{c}_{i^*})$ from  the linear system \eqref{eq1}. Obviously, a necessary condition for the repair is ${\rm rank}(W_{i^*})=\ell$, or equivalently, $\dim_B\big(\{\textbf{g}_{i^*}^{(j)}\}_{j=1}^\ell\big)=\ell$. Moreover, terms on the right hand of  \eqref{eq1} should be obtained from the helper nodes during the  repair process. Thus, the helper node $i$  for $i\neq i^*$ needs to access ${\rm nz}(W_i)$ subsymbols and transmit only ${\rm rank}(W_i)$ subsymbols to node $i^*$, where ${\rm nz}(W_i)$ denotes the number of nonzero columns in $W_i$.  Therefore,  the I/O cost of the repair scheme is
\begin{equation}\label{eq3}
    \gamma_{I/O}=\sum_{i\in[n]\setminus\{i^*\}}{\rm nz}(W_i)=\sum_{i\in[n]}{\rm nz}(W_i)-\ell,
\end{equation}
where the last equality follows from the necessary condition ${\rm nz}(W_{i^*})={\rm rank}(W_{i^*})=\ell$ for the repair process.
Meanwhile, the repair bandwidth of the repair scheme is
\begin{equation}\label{eq4}
    b=\sum_{i\in[n]\setminus\{i^*\}}{\rm rank}(W_i)=\sum_{i\in[n]}{\rm rank}(W_i)-\ell.
\end{equation}

Formally, the linear repair scheme for scalar MDS codes over $F$ is defined as follows.
\begin{definition}[linear repair scheme] Let $\mathcal{C}$ be an $[n,k]$ MDS code over $F$, and $B\subseteq F$ be a subfield of $F$ with $[F:B]=\ell$. A linear repair scheme  for node $i^*$ over $B$ is characterized by $\ell$ dual codewords ${\bm g}^{(1)},...,{\bm g}^{(\ell)}\in\mathcal{C}^{\bot}$, where ${\bm g}^{(j)}=(\textbf{g}_1^{(j)},...,\textbf{g}_n^{(j)})$, satisfying $\dim_B\big(\{\textbf{g}_{i^*}^{(j)}\}_{j=1}^\ell\big)=\ell$. Moreover, the repair bandwidth is $b=\sum_{i\in[n]}{\rm rank}(W_i)-\ell$ and the I/O cost with respect to $\mathcal{B}$ is $\gamma_{I/O}=\sum_{i\in[n]}{\rm nz}(W_i)-\ell$, where $W_i$ is defined as in (\ref{W_j}).
\end{definition}
\begin{remark}The I/O cost depends on the choice of basis $\mathcal{B}$. That is, the same repair scheme (i.e., ${\bm g}^{(1)},...,{\bm g}^{(\ell)}\in\mathcal{C}^{\bot}$) may have different I/O cost with respect to different bases. However, the repair bandwidth is totally determined by the repair scheme, because ${\rm rank}(W_i)$ remains unchanged under different bases.
\end{remark}

In \cite{I/OFormula}, the authors transformed the problem of calculating the I/O cost to the calculation of the Hamming weight of some linear space by using the following lemma.
\begin{lemma}\label{wt}
	Let $G$ be a $k\times m$ matrix over $B$. Then,
    \begin{equation*}
		{\rm nz}(G)=\frac{1}{q^{k-1}(q-1)}\sum_{{\bm u}\in B^k}{\rm wt}({\bm u}G).
    \end{equation*}	
\end{lemma}
Note that the rows of $G$ in \citep[Lemma 4]{I/OFormula} are linearly independent over $B$. Here we extend this result to general $k\times m$ matrices  using a similar proof.

\begin{lemma}[\cite{I/OFormula}]\label{I/Oformula}
Suppose $\{{\bm g}^{(j)}\}_{j=1}^\ell\subseteq \mathcal{C}^\bot$ defines a repair scheme for node $i^*$. Then the I/O cost of the repair scheme with respect to $\mathcal{B}$ is
	$$\gamma_{I/O}=\frac{\sum_{{\bm u}\in B^\ell}{\rm wt}({\bm u}G_{i^*})}{q^{\ell-1}(q-1)}-\ell,$$
where $G_{i^*}=(W_1~W_2~\cdots~W_n)$ and $W_i, i\in[n]$, is defined as in (\ref{W_j}).
\end{lemma}

\subsection{Additive character}	
An additive character $\chi$ of a finite abelian group $G$ is a homomorphism from $G$ into the multiplicative group $U$ of complex numbers of absolute value 1. For any $a \in \mathbb{F}_{q^t}$, the function
$$
\chi_a(x)=\zeta_p^{{\rm Tr}(a x)},
$$
defines an additive character of $\mathbb{F}_{q^t}$, where $\zeta_p=e^{\frac{2 \pi \sqrt{-1}}{p}}$ and $p={\rm Char}(\mathbb{F}_{q^t})$. When $a=0$, $\chi_0(x)=1$ for all $x \in \mathbb{F}_{q^t}$, which is called the trivial additive character of $\mathbb{F}_{q^t}$. When $a=1$,  $\chi_1(x)=\zeta_p^{{\rm Tr}(x)}$, which is called the canonical additive character of $\mathbb{F}_{q^t}$.
\begin{lemma}\citep[Theorem 5.4]{finite field}\label{thm4}	
If $\chi$ is a nontrivial character of the finite abelian group $G$, then $\sum_{g\in G}\chi(g)=0$.
\end{lemma}
\begin{corollary}\label{char} Suppose $G$ is a $B$-linear subspace of $F$ and $\chi$ is the canonical additive character of $F$,  then	
\begin{equation*}
	\sum_{\alpha\in G}\chi(\alpha)=\begin{cases}
		|G| &{\rm if}~ G\subseteq{\rm Ker}\big({\rm Tr}_{F/B}\big),\\
		0	&{\rm otherwise}.\\
	\end{cases}
\end{equation*}
\end{corollary}
\begin{proof}
Since $G\subseteq F$ is a finite abelian group, $\chi|_{G}$ is also a character of $G$. Moreover, $\chi|_{G}$ is a trivial character of $G$ if and only if $\chi(\alpha)=1$ for every $\alpha\in G$, i.e., $G\subseteq{\rm Ker}\big({\rm Tr}\big)$. Combining with Lemma \ref{thm4}, it holds that
\begin{equation*}
	\sum_{\alpha\in G}\chi(\alpha)=\begin{cases}
		|G| &{\rm if}~ G\subseteq{\rm Ker}\big({\rm Tr}\big),\\
		0	&{\rm otherwise}.\\
	\end{cases}
\end{equation*}
Then, it is sufficient to show that  $G\subseteq{\rm Ker}\big({\rm Tr}\big)$ if and only if $G\subseteq{\rm Ker}\big({\rm Tr}_{F/B}\big)$. Note that ${\rm Ker}\big({\rm Tr}_{F/B}\big)\subseteq{\rm Ker}\big({\rm Tr}\big)$, we only need to prove the necessity. Assume $G\subseteq{\rm Ker}\big({\rm Tr}\big)$. Then, for any $\alpha\in G$ and $c\in B$, $c\alpha\in G\subseteq{\rm Ker}\big({\rm Tr}\big)$ because $G$ is a $B$-linear subspace. Therefore,  for all  $c\in B$, it has ${\rm Tr}_{B/\mathbb{F}_p}(c{\rm Tr}_{F/B}(\alpha))={\rm Tr}(c\alpha)=0$, which implies  ${\rm Tr}_{F/B}(\alpha)=0$. That completes the proof.
\end{proof}

\section{Calculating I/O cost via exponential sums}\label{Sec3}
In this section, we initially define the normalized polynomials, which simplify the calculation of I/O cost and repair bandwidth of linear repair schemes. Then, we introduce additive characters to calculate the I/O cost and derive a concise formula for the I/O cost. By applying this formula in conjunction with the Weil bound, we establish a lower bound on the I/O cost for full-length RS codes over $F$ with $r\leq{\rm Char}(F)$ parities.

\subsection{The $(m,t)$-normalized polynomials}
First recall some basics of RS codes.
Let $\mathcal{A}=\left\{\alpha_1, \alpha_2 ,...,\alpha_n \right\}\subseteq F$. The $[n,k]$ RS code over $F$ with the evaluation points set $\mathcal{A}$ is defined as
\begin{equation*}
	{\rm RS}(\mathcal{A},k)=\{(f(\alpha_1),...,f(\alpha_n)):f\in F[x], \mathrm{deg}(f)\leq k-1\}.
\end{equation*}
Hereafter, we always assume $\mathcal{A}$ is a $B$-linear subspace of $F$. It is known in this case ${\rm RS}(\mathcal{A},k)^\bot={\rm RS}(\mathcal{A},n-k)$.
Therefore, any linear repair scheme for node $i^*$ of ${\rm RS}(\mathcal{A},k)$ corresponds to $\ell$ polynomials $g_j(x),j\in[\ell]$, of degree less than $n-k$ over $F$, such that $\dim_B\big(\{g_j(\alpha_{i^*})\}_{j=1}^\ell\big)=\ell$.  Next, we are to derive a normalized form of the repair schemes which will be used to simplify the later proofs. The normalization is based on the following observation.

\begin{proposition}\label{equivscheme}
Suppose $\{g_j(x)\}_{j=1}^\ell$ defines a linear repair scheme for  node $i^*$ of ${\rm RS}(\mathcal{A},k)$. Let $M\in B^{\ell\times\ell}$ be an invertible matrix and define
\begin{equation}\label{eq8}
	\begin{pmatrix}
		\tilde{g}_1(x)\\
		\tilde{g}_2(x)\\
		\vdots\\
		\tilde{g}_{\ell}(x)\\
	\end{pmatrix}=M\begin{pmatrix}
	g_1(x)  \\
	g_2(x) \\
	\vdots\\
	g_{\ell}(x)\\
	\end{pmatrix}.
\end{equation}
Then, $\{\tilde{g}_j(x)\}_{j=1}^\ell$ also characterizes a repair scheme for node $i^*$ of ${\rm RS}(\mathcal{A},k)$. Moreover, the two schemes have the same I/O cost and repair bandwidth.
\end{proposition}

\begin{proof}
Denote $$W_i\!=\!\begin{pmatrix}
	\Phi_{\hat{\mathcal{B}}}(g_1(\alpha_i))  \\
	\Phi_{\hat{\mathcal{B}}}(g_2(\alpha_i)) \\
	\vdots\\
	\Phi_{\hat{\mathcal{B}}}(g_{\ell}(\alpha_i))
\end{pmatrix}, \;\;\tilde{W}_i\!=\!\begin{pmatrix}
	\Phi_{\hat{\mathcal{B}}}(\tilde{g}_1(\alpha_i))  \\
	\Phi_{\hat{\mathcal{B}}}(\tilde{g}_2(\alpha_i)) \\
	\vdots\\
	\Phi_{\hat{\mathcal{B}}}(\tilde{g}_{\ell}(\alpha_i))
\end{pmatrix}.$$
From (\ref{eq8}) and the $B$-linearity of the map $\Phi_{\hat{\mathcal{B}}}$, it has $\tilde{W}_i=MW_i$ for all $i\in[n]$. Moreover, it is easy to see that ${\rm nz}(\tilde{W}_i)={\rm nz}(W_i)$ and ${\rm rank}(\tilde{W}_i)={\rm rank}(W_i)$, because they only differ by elementary row transformations. Then the proposition follows from (\ref{eq1})-(\ref{eq4}).
\end{proof}

For simplicity, we call the two repair schemes $\{g_j(x)\}_{j=1}^\ell$ and $\{\tilde{g}_j(x)\}_{j=1}^\ell$ equivalent if they differ by an invertible $B$-linear transformation as in Proposition \ref{equivscheme}. Furthermore, we define the normalized repair scheme to be the equivalent scheme that contains the maximum number of constant polynomials.

\begin{definition}\label{def}
Suppose $\{g_j(x)\}_{j=1}^\ell$ defines a linear repair scheme for  node $i^*$ of ${\rm RS}(\mathcal{A},k)$. We say
$\{g_j(x)\}_{j=1}^\ell$ is  $(m,t)$-normalized with respect to $\mathcal{B}$ if the following two conditions hold:
\begin{enumerate}
\item ${\rm deg}(g_{\bm u})>0$ for any ${\bm u}\in B^m\setminus\{{\bm 0}\}$, where $g_{\bm u}(x)=\sum_{j=1}^m u_jg_j(x)$ for  ${\bm u}=(u_1,...,u_m)$.

\item $g_j(x)=\omega_j$ for $j\in[m+1,\ell]$, and $\big|\bigcup_{j=m+1}^{\ell}{\rm supp}(\Phi_{\hat{\mathcal{B}}}(\omega_j))\big|=\ell-t$.
\end{enumerate}
\end{definition}

In other words, for an $(m,t)$-normalized repair scheme, the last $\ell-m$ polynomials are all constant polynomials. Moreover, no more constant polynomials can be added under the equivalence sense due to condition 1). The second half of condition 2) just indicates the support size of all the constants under the map $\Phi_{\hat{\mathcal{B}}}$.

\begin{remark}
For an $(m,t)$-normalized repair scheme $\{g_j(x)\}_{j=1}^\ell$ of node $i^*$, it follows $\dim_B\big(\{g_j(\alpha_{i^*})\}_{j=1}^\ell\big)=\ell$. Subsequently, $\dim_B\big(\{\omega_j\}_{j=m+1}^\ell\big)=\ell-m$. Combining with $\dim_B\big(\{\omega_j\}_{j=m+1}^\ell\big)\leq \big|\bigcup_{j=m+1}^{\ell}{\rm supp}(\Phi_{\hat{\mathcal{B}}}(\omega_j))\big|=\ell-t$, it always has $t\leq m$.
\end{remark}

\begin{theorem}\label{thm7}
Any linear repair scheme for node $i^*$ of ${\rm RS}(\mathcal{A},k)$ is equivalent to some $(m,t)$-normalized repair scheme $\{g_j(x)\}_{j=1}^\ell$ with respect to $\mathcal{B}$.
\end{theorem}
\begin{proof}
For any repair scheme $\{\tilde{g}_j(x)\}_{j=1}^{\ell}$, define $U=\{{\bm u}\in B^\ell:{\rm deg}(\tilde{g}_{\bm u})= 0\}$ where $\tilde{g}_{\bm u}(x)=\sum_{j=1}^\ell u_j\tilde{g}_j(x)$ for ${\bm u}=(u_1,...,u_m)$. It can be seen that $U$ is a linear subspace of  $B^\ell$.  Assume $\dim_B(U)=\ell-m$ where $m\in[0,\ell]$. Let $\{{\bm u}_{m+1},{\bm u}_{m+2},...,{\bm u}_{\ell}\}$ be a basis of $U$, then extend it to a basis of $B^\ell$, $\{{\bm u}_{1},...,{\bm u}_{\ell}\}$. Thus, let $$\begin{pmatrix}
	g_1(x)  \\
	\vdots\\
	g_{\ell}(x)\\
	\end{pmatrix}=\begin{pmatrix}
  {\bm u}_1\\\vdots\\{\bm u}_{\ell}
\end{pmatrix}\begin{pmatrix}
		\tilde{g}_1(x)\\
		\vdots\\
		\tilde{g}_{\ell}(x)\\
	\end{pmatrix}$$ which is the normalized equivalent repair scheme.
\end{proof}
By Theorem \ref{thm7} and Proposition \ref{equivscheme}, it suffices to consider only normalized repair schemes when computing the I/O cost and repair bandwidth.
Let $\{g_{j}(x)\}_{j=1}^\ell$ be an $(m,t)$-normalized repair scheme with respect to $\mathcal{B}$ for  node $i^*$. Moreover, we assume without loss of generality that $\bigcup_{j=m+1}^{\ell}{\rm supp}(\Phi_{\hat{\mathcal{B}}}(\omega_j))=[t+1,\ell]$. Thus, the repair matrix $W_i$ defined in (\ref{W_j}) has a special form, i.e.,
\begin{equation}\label{IOmatrix}
\fontsize{8pt}{2pt}
            {W_i=\left({\begin{array} {cccc|ccc}
			{\rm Tr}_{F/B}(g_1(\alpha_i)\beta^{(1)})&{\rm Tr}_{F/B}(g_1(\alpha_i)\beta^{(2)}) &\cdots &{\rm Tr}_{F/B}(g_1(\alpha_i)\beta^{(t)}) & * &\cdots &*\\
			{\rm Tr}_{F/B}(g_2(\alpha_i)\beta^{(1)})&{\rm Tr}_{F/B}(g_2(\alpha_i)\beta^{(2)}) &\cdots &{\rm Tr}_{F/B}(g_2(\alpha_i)\beta^{(t)}) & * &\cdots &*\\
			\vdots & \vdots &\ddots &\vdots & \vdots &\ddots & \vdots \\
			{\rm Tr}_{F/B}(g_m(\alpha_i)\beta^{(1)})&{\rm Tr}_{F/B}(g_m(\alpha_i)\beta^{(2)}) &\cdots &{\rm Tr}_{F/B}(g_m(\alpha_i)\beta^{(t)}) & * &\cdots &*\\
			\hline
			0 & 0&\cdots & 0 &{\rm Tr}_{F/B}(\omega_{m+1}\beta^{(t+1)}) &\cdots &{\rm Tr}_{F/B}(\omega_{m+1}\beta^{(\ell)})\\
			\vdots&\vdots &\ddots &\vdots & \vdots &\ddots & \vdots \\
			0 & 0&\cdots & 0 &{\rm Tr}_{F/B}(\omega_{\ell}\beta^{(t+1)}) &\cdots &{\rm Tr}_{F/B}(\omega_{\ell}\beta^{(\ell)})\\
			\end {array}}\right).}
	\end{equation}
Denote the upper left corner of $W_i$ in (\ref{IOmatrix}) by $\hat{W}_i$, then ${\rm nz}(W_i)={\rm nz}(\hat{W}_i)+\ell-t$. Combining with (\ref{eq3}), the I/O cost of the repair scheme with respect to $\mathcal{B}$ is
\begin{align}
\gamma_{I/O}&=(n-1)\ell-nt+\sum_{i=1}^n{\rm nz}(\hat{W}_i)\label{7}\\
&=(n-1)\ell-nt+\frac{1}{q^{m-1}(q-1)}\sum_{i=1}^n\sum_{{\bm u}\in B^m}{\rm wt}({\bm u}\hat{W}_i)\label{eq9+},
\end{align}
where \eqref{eq9+} follows from Lemma \ref{wt}. As for the repair bandwidth, it has $b=\sum_{i=1}^n{\rm rank}(W_i)-\ell$. Particularly when $t=m$, the lower right corner of $W_i$ is invertible because $\dim_B\big(\{\omega_{m+1},...,\omega_{\ell}\}\big)=\ell-m$. Then the upper right corner of $W_i$ can become all zeros by elementary row transformations while the upper left corner remains unchanged. That is, $W_i$ can be transformed into a block diagonal matrix when $t=m$, and thus ${\rm rank}(W_i)={\rm rank}(\hat{W_i})+\ell-m$. As a result, the repair bandwidth can be calculated as
\begin{equation}\label{b-formula}
b=\sum_{i=1}^n{\rm rank}(W_i)-\ell=(n-1)\ell-nm+\sum_{i=1}^n{\rm rank}(\hat{W}_i).
\end{equation}
Therefore, instead of considering the $\ell\times\ell$ matrix $W_i$ for general linear repair schemes, we can only focus on the $m\times t$ matrix $\hat{W}_i$ for $(m,t)$-normalized schemes, which simplify the calculation of the I/O cost and repair bandwidth.


\subsection{Calculating the I/O cost via exponential sums}\label{uW_i}
Combining with the relation between the Hamming weight of ${\bm u}\hat{W}_i$ and the additive character sum, we can derive a concise formula for the  I/O cost as follows.

\begin{theorem}\label{thmformula}
Suppose $\chi$ is the canonical additive character of $F$.
Let $\{g_j(x)\}_{j=1}^\ell$ be an $(m,t)$-normalized repair scheme with respect to $\mathcal{B}$ for node $i^*$ of ${\rm RS}(\mathcal{A},k)$. The I/O cost of the scheme  with respect to $\mathcal{B}$ is	
\begin{equation*}\label{keyeq1}
	\gamma_{I/O}=(n-1)\ell-\frac{1}{q^m}\sum_{s=1}^t\sum_{{\bm u}\in B^m}\sum_{\alpha\in\mathcal{A}}\chi(g_{{\bm u}}(\alpha)\beta^{(s)}),
\end{equation*}
where $g_{\bm u}(x)=\sum_{j=1}^mu_jg_j(x)$ for all ${\bm u}=(u_1,...,u_m)\in B^m$.	
\end{theorem}	

\begin{proof}
By \eqref{eq9+}, we need to calculate ${\rm wt}({\bm u}\hat{W}_i)$ for $i\in[n]$ and ${\bm u}\in B^m$, where
\begin{equation*}\label{eq10}
	{\bm u}\hat{W}_i=\Big({\rm Tr_{F/B}}\big(g_{\bm u}(\alpha_i)\beta^{(1)}\big),{\rm Tr_{F/B}}\big(g_{\bm u}(\alpha_i)\beta^{(2)}\big),...,{\rm Tr_{F/B}}\big(g_{\bm u}(\alpha_i)\beta^{(t)}\big)\Big).
\end{equation*}
In order to count the nonzeros, we introduce the additive character.
Let $\hat{\chi}$ be  the canonical additive character of $B$, then $\chi(\alpha)=\hat{\chi}({\rm Tr}_{F/B}(\alpha))$ for any $\alpha\in F$. It can be seen that
\begin{align}
	\sum_{z\in B}\chi(zg_{{\bm u}}(\alpha_i)\beta^{(s)})&=\sum_{z\in B}\hat{\chi}\big({\rm Tr}_{F/B}(zg_{{\bm u}}(\alpha_i)\beta^{(s)})\big)\notag\\
	&=\sum_{z\in B}\hat{\chi}\big(z{\rm Tr}_{F/B}(g_{{\bm u}}(\alpha_i)\beta^{(s)})\big)\notag\\
	&=\begin{cases}
		q &{\rm if}~ {\rm Tr}_{F/B}(g_{\bm u}(\alpha_i)\beta^{(s)})=0\\
		0 &{\rm if}~ {\rm Tr}_{F/B}(g_{\bm u}(\alpha_i)\beta^{(s)})\neq0\\
	\end{cases}\;,\label{eq11}
\end{align}
where equality (\ref{eq11}) is based on Lemma \ref{thm4} and the fact $\{z{\rm Tr}_{F/B}(g_{{\bm u}}(\alpha_i)\beta^{(s)}):z\in B\}=B$ when ${\rm Tr}_{F/B}(g_{\bm u}(\alpha_i)\beta^{(s)})\neq0$.
Therefore, for any fixed ${\bm u}\in B^m$ and $i\in[n]$, as $s$ ranges in $[t]$, ${\rm Tr}_{F/B}(g_{\bm u}(\alpha_i)\beta^{(s)})=0$ if and only if $\sum_{z\in B}\chi(zg_{{\bm u}}(\alpha_i)\beta^{(s)})=q$.
Thus we can write ${\rm wt}({\bm u}\hat{W}_i)=t-\frac{1}{q}\sum_{s=1}^t\sum_{z\in B}\chi(zg_{{\bm u}}(\alpha_i)\beta^{(s)})$.

Therefore,
\begin{align}
	\sum_{i=1}^n\sum_{{\bm u}\in B^m}({\rm wt}({\bm u}\hat{W}_i))&=ntq^m-\frac{1}{q}\sum_{i=1}^n\sum_{{\bm u}\in B^m}\sum_{s=1}^t\sum_{z\in B}\chi(zg_{{\bm u}}(\alpha_i)\beta^{(s)})\notag\\
	&=ntq^m-ntq^{m-1}-\frac{1}{q}\sum_{z\in B^*}\Big(\sum_{i=1}^n\sum_{s=1}^t\sum_{{\bm u}\in B^m}\chi(zg_{{\bm u}}(\alpha_i)\beta^{(s)})\Big)\label{eq111}\\
	&=nt(q^m-q^{m-1})-\frac{q-1}{q}\sum_{i=1}^n\sum_{s=1}^t\sum_{{\bm u}\in B^m}\chi(g_{{\bm u}}(\alpha_i)\beta^{(s)}),\label{eq13}
\end{align}	
where (\ref{eq111}) is because $\chi(zg_{{\bm u}}(\alpha_i)\beta^{(s)})=1$ as $z=0$, and equality (\ref{eq13}) is due to $zg_{{\bm u}}(x)=g_{z{\bm u}}(x)$ and $z{\bm u}$ also runs over $B^m$ while ${\bm u}$ ranges in $B^m$ and $z$ is fixed in $B^*\triangleq B\setminus\{0\}$.  Combining the equalities (\ref{eq9+}) and (\ref{eq13}), the proof completes.
\end{proof}

This theorem relates the I/O cost to exponential sums which have been widely studied and broadly applied in information theory. Furthermore, by using the Weil bound for exponential sums, we can derive the following lower bound.
\begin{corollary}\label{coro11}
	For ${\rm RS}(F,n-r)$ over $F$ with $r\leq {\rm Char}(F)$, the I/O cost of any linear repair scheme satisfies:
	\begin{equation*}
		\gamma_{I/O}\geq(n-1)\ell-q^{\ell-1}-(r-2)(q-1)q^{\frac{\ell}{2}-1},
	\end{equation*}
    where $n=q^\ell$.	
\end{corollary}
\begin{proof}
 The proof is presented in Appendix \ref{proof-coro11}.
\end{proof}
\begin{remark}
Note when $r=2$, the lower bound obtained in Corollary \ref{coro11} coincides with the bound derived in \cite{I/OFormula} for full-length RS codes with two parities. Moreover, the bound is tight in this case due to the scheme constructed there.
\end{remark}

\section{Lower bounds on I/O for RS codes evaluated on subspaces }\label{Sec4}
In this section, we characterize the I/O cost of linear repair schemes for RS codes ${\rm RS}(\mathcal{A},n-r)$ over $F$, where  $\mathcal{A}$ is a $d$-dimensional $B$-linear subspace of $F$ and $n=q^d$. Moreover, we determine the repair bandwidth for the repair schemes that achieve the optimal I/O cost. For simplicity, denote $\mathcal{A}=\{\alpha_1=0,\alpha_2,...,\alpha_n\}\subseteq F$.
Similar to \citep[Lemma 8]{IOfulllength}, the repair scheme for each node of ${\rm RS}(\mathcal{A},k)$ has the same repair bandwidth and I/O cost at all nodes when $\mathcal{A}$ is a $d$-dimensional $B$-linear subspace of $F$, so it suffices to just examine the repair scheme of node $0$ (corresponding to the evaluation point $\alpha_1=0$).

Denote $K={\rm Ker}({\rm Tr}_{F/B})$. We first recall a
lemma from \citep[Lemma 5]{IOfulllength} which will be used in later proofs.

\begin{lemma}\label{lem10}(\cite{IOfulllength}) Suppose $\beta_1,...,\beta_t\in F\setminus\{0\}$, then $\mathrm{dim}_B\big(\bigcap_{j=1}^t\beta_{j}^{-1}K\big)=\ell-\dim_B\big(\{\beta_j\}_{j=1}^t\big)$. 
\end{lemma}
\subsection{RS codes with two parities}
We first derive the lower bound on the I/O cost for RS codes with two parities.
\begin{theorem}\label{r=2} Let $\mathcal{A}$ be a $d$-dimensional $B$-linear subspace of $F$. Then
	the I/O cost of any linear repair scheme for ${\rm RS}(\mathcal{A},q^d-2)$ over $F$ satisfies:
	\begin{equation}\label{r2}
		\gamma_{I/O} \geq (n-1)\ell-(\ell-d+1)q^{d-1},
	\end{equation}where $n=q^d$.
\end{theorem}
\begin{proof}
Let $\mathcal{B}=\{\beta^{(1)},...,\beta^{(\ell)}\}$ be a basis of $F$ over $B$. Suppose $\{g_j(x)\}_{j=1}^\ell$ is an $(m,t)$-normalized repair scheme for node $\alpha_1=0$ with respect to $\mathcal{B}$. Since ${\rm RS}(\mathcal{A},q^d-2)^\bot={\rm RS}(\mathcal{A},2)$, we have ${\rm deg}(g_j)\leq 1$ for $j\in[\ell]$. Moreover, from Definition \ref{def} we may assume
\begin{equation}\label{eq15}
	g_j(x)=\begin{cases}
		\eta_jx+\omega_j &{\rm if}~ j\in[m]\\
		\omega_j	&{\rm if}~ j\in[m+1,\ell]\\
	\end{cases}\;,
\end{equation}
where $\eta_j,\omega_j\in F$, $\bigcup_{j=m+1}^{\ell}{\rm supp}(\Phi_{\hat{\mathcal{B}}}(\omega_j))=[t+1,\ell]$, and $\dim_B(\{\omega_j\}_{j=1}^\ell)=\ell$. Denote $g_{{\bm u}}(x)=\sum_{j=1}^m u_jg_j(x)=\eta_{{\bm u}}x+\omega_{{\bm u}}$ for ${\bm u}=(u_1,...,u_m)\in B^m$, where $\eta_{{\bm u}}=\sum_{j=1}^m u_j\eta_j, ~\omega_{{\bm u}}=\sum_{j=1}^m u_j\omega_j$. Due to $1)$ of Definition \ref{def}, we know ${\rm deg}(g_{\bm u})>0$ for all ${\bm u}\in B^m\setminus\{{\bm 0}\}$, or, equivalently, $\dim_B\big(\{\eta_j\}_{j=1}^m\big)=m$.

By Theorem \ref{thmformula}, the I/O cost of $\{g_j(x)\}_{j=1}^\ell$ with respect to $\mathcal{B}$ is
\begin{align}\label{eq14}
		\gamma_{I/O}=(n-1)\ell-\frac{1}{q^m}\sum_{s=1}^t\sum_{{\bm u}\in B^m}\sum_{\alpha\in\mathcal{A}}\chi(g_{{\bm u}}(\alpha)\beta^{(s)}),
\end{align} where $n=q^d$ and $\chi$ is the canonical additive character of $F$.
First, we calculate
\begin{align}
	\sum_{\alpha\in\mathcal{A}}\chi(g_{{\bm u}}(\alpha)\beta^{(s)})&=\sum_{\alpha\in\mathcal{A}}\chi(\eta_{{\bm u}}\alpha\beta^{(s)}+\omega_{{\bm u}}\beta^{(s)})\notag\\
		                &\notag=\chi(\omega_{{\bm u}}\beta^{(s)})\sum_{\alpha\in\mathcal{A}}\chi(\eta_{{\bm u}}\beta^{(s)}\alpha)
\end{align}
where the second equality is due to the homomorphism property of additive characters.
Then it follows from Corollary \ref{char} that
\begin{equation*}
	\sum_{\alpha\in\mathcal{A}}\chi(\eta_{{\bm u}}\beta^{(s)}\alpha)=\begin{cases}
		q^d &{\rm if}~ \eta_{{\bm u}}\beta^{(s)}\mathcal{A}\subseteq K ,\\
		0  &{\rm otherwise}.
	\end{cases}
\end{equation*}
For $s\in[t]$, denote $U^{(s)}=\{{\bm u}\in B^m:\eta_{{\bm u}}\beta^{(s)}\mathcal{A}\subseteq K\}$ and $W^{(s)}=\{\omega_{\bm u}:{\bm u}\in U^{(s)}\}$. It can be seen that $U^{(s)}$ is a $B$-linear subspace of $B^m$ and $W^{(s)}$ is a $B$-linear subspace of $F$. Moreover, $W^{(s)}$ and $U^{(s)}$ have the same dimension because $\dim_B(\{\omega_j\}_{j=1}^m)=m$. Assume $\dim_B(U^{(s)})=\dim_B(W^{(s)})=a_s$ for $s\in[t]$.
Then, (\ref{eq14}) becomes
\begin{align}
	\gamma_{I/O}&=(n-1)\ell-q^{d-m}\sum_{s=1}^t\sum_{{\bm u}\in U^{(s)}}\chi(\omega_{{\bm u}}\beta^{(s)})\notag\\
		        &=(n-1)\ell-q^{d-m}\sum_{s=1}^t\sum_{\omega\in W^{(s)}}\chi(\omega\beta^{(s)}).\label{eq26}
\end{align}
Again, by Corollary \ref{char} it has
\begin{equation*}
	\sum_{\omega\in W^{(s)}}\chi(\omega\beta^{(s)})=\begin{cases}
		q^{a_s} &{\rm if}~ \beta^{(s)}W^{(s)}\subseteq K,\\
		0       &{\rm otherwise}.\\
	\end{cases}
\end{equation*}
Denote  $t'=|\{s\in[t]: \beta^{(s)}W^{(s)}\subseteq K\}|$. If $t'=0$, $\gamma_{I/O}=(n-1)\ell$ and thus Theorem \ref{r=2} obviously holds. If $t'>0$, without loss of generality we assume $\{s\in[t]: \beta^{(s)}W^{(s)}\subseteq K\}=[t']$. Then (\ref{eq26}) becomes
\begin{equation}\label{eq27+}
	\gamma_{I/O}=(n-1)\ell-q^{d-m}\sum_{s=1}^{t'}q^{a_s}.
\end{equation}
We have the following claim to establish some restrictions on $\{a_s\}_{s=1}^{t'}$ which are then used to estimate $\gamma_{I/O}$.

{\bf Claim.} $a_s\leq\min\{\ell-d,m-1\}$ for $s\in [t']$.

{\it Proof of the claim. }We first prove $a_s\leq\ell-d$. Let $\{{\bm u}_1,...,{\bm u}_{a_s}\}$ be a basis of $U^{(s)}$. By the definition of $U^{(s)}$, it holds $\beta^{(s)}\mathcal{A}\subseteq\bigcap_{j=1}^{a_s}\eta_{{\bm u}_j}^{-1}K$. Since $\dim_B(\{\eta_j\}_{j=1}^m)=m$, it follows $\dim_B\big(\{\eta_{{\bm u}_j}\}_{j=1}^{a_s}\big)=\dim_B\big(U^{(s)}\big)=a_s$. Then by Lemma \ref{lem10} it has $\dim_B\big(\bigcap_{j=1}^{a_s}\eta_{{\bm u}_j}^{-1}K\big)=\ell-a_s$. Therefore, $d=\dim_B(\beta^{(s)}\mathcal{A})\leq \dim_B\big(\bigcap_{j=1}^{a_s}\eta_{{\bm u}_j}^{-1}K)\big)=\ell-a_s$, i.e., $a_s\leq\ell-d$. Next, we prove $a_s\leq m-1$.
Since $s\not\in\bigcup_{j=m+1}^{\ell}{\rm supp}(\Phi_{\hat{\mathcal{B}}}(\omega_j))=[t+1,\ell]$ for $s\leq t'\leq t$, it has ${\rm Tr}_{F/B}(\omega_j\beta^{(s)})=0$ for any $j\in[m+1,\ell]$. That is,
${\rm span}_B(\{\omega_{m+1},...,\omega_{\ell}\})\subseteq\big(\beta^{(s)}\big)^{-1}K$. Recall that $W^{(s)}\subseteq{\rm span}_B(\{\omega_1,...,\omega_m\})$ and $\dim_B\big(\{\omega_j\}_{j=1}^\ell\big)=\ell$, hence $W^{(s)}\cap{\rm span}_B(\{\omega_{m+1},...,\omega_{\ell}\})=\{0\}$. Then, $W^{(s)}\oplus{\rm span}_B(\{\omega_{m+1},...,\omega_{\ell}\})\subseteq\big(\beta^{(s)}\big)^{-1}K$. Thus, $a_s+\ell-m\leq \ell-1$, i.e., $a_s\leq m-1$.

Combining the Claim and (\ref{eq27+}), it has
\begin{align}
	\gamma_{I/O}&\geq (n-1)\ell-q^{d-m}\sum_{s=1}^{t'}q^{\min\{\ell-d,m-1\}}\notag\\
	            &=(n-1)\ell-t'q^{\min\{\ell-m,d-1\}}\notag\\
	            &\geq(n-1)\ell-mq^{\min\{\ell-m,d-1\}}\label{eq16} \\
              &\geq(n-1)\ell-(\ell-d+1)q^{d-1},\label{eq17}
\end{align}
where (\ref{eq16}) follows from  $t'\leq t\leq m$, and (\ref{eq17}) is due to the fact that $mq^{\min\{\ell-m,d-1\}}$ reaches the maximal at $m=\ell-d+1$ as $m$ ranges in $[\ell]$.
\end{proof}

Next, we further give a necessary condition for the lower bound (\ref{r2}) to hold with equality.

\begin{corollary}\label{condr=2}
Let $\{g_j(x)\}_{j=1}^\ell$ be a repair scheme of node $0$ for ${\rm RS}(\mathcal{A},q^d-2)$. If $\{g_j(x)\}_{j=1}^\ell$ is  $(m,t)$-normalized with respect to $\mathcal{B}$ and meets the lower bound (\ref{r2}) with equality, it must have $t=m$, and \begin{itemize}
  \item[(i)] for $d=\ell$, $m=1$ when $q>2$, or, $m\in\{1,2\}$ when $q=2$;
  \item[(ii)] for $d<\ell$, $m=\ell-d+1$.
\end{itemize}
\end{corollary}

\begin{proof}
From (\ref{eq16}), we know the equality holds when $a_s=\min\{\ell-d,m-1\}$ and $t'=m$. Because $t'\leq t\leq m$, it must have $t=m$. Then based on (\ref{eq17}), it can be easily checked that $mq^{\min\{\ell-m,d-1\}}=(\ell-d+1)q^{d-1}$ implies (i) for $d=\ell$, and  (ii) for $d<\ell$.
\end{proof}

Actually, the lower bound (\ref{r2}) is tight when $\ell-d+1\mid\ell$. Specifically, for $d=\ell$, the bound coincides with the bound derived in \cite{fullr=2} and \cite{I/OFormula}, where the tightness has been shown. For the case $d<\ell$ and $\ell-d+1\mid\ell$, the bound is achieved by the scheme built in {\bf Construction 2} in Section \ref{Sec5}. We call the repair scheme with the I/O cost matching the bound (\ref{r2}) as an I/O-optimal repair scheme. In the following, we characterize the repair bandwidth of the I/O-optimal linear repair scheme in the case of $\ell-d+1\mid\ell$.

\begin{theorem}\label{b-r=2-d}
Suppose $\mathcal{A}$ is a $d$-dimensional $B$-linear subspace of $F$ and $\ell-d+1|\ell$. The repair bandwidth, denoted by $b$, of any I/O-optimal linear repair scheme for ${\rm RS}(\mathcal{A},q^d-2)$ satisfies:
\begin{itemize}
\item [{\rm(i)}] If $d=\ell$ and $q>2$, $b=(q^\ell-1)\ell-q^{\ell-1}$.
\item[{\rm(ii)}] If $d=\ell$ and $q=2$, $b\geq (2^\ell-1)\ell-3\cdot 2^{\ell-2}$.
\item[{\rm(iii)}] If $d<\ell$, $b\geq (q^d-1)d-q^{2d-\ell-1}$.
\end{itemize}
\end{theorem}
\begin{proof}
Suppose $\{g_j(x)\}_{j=1}^\ell$ defines a repair scheme for node $\alpha_1=0$ which is $(m,t)$-normalized  with respect to $\mathcal{B}$ and I/O-optimal.
We may assume $g_j(x)$ is defined as in (\ref{eq15}) for $j\in[\ell]$. By Corollary \ref{condr=2}, it must have $m=t$. Combining with (\ref{b-formula}), the repair bandwidth relies on the computation of $\sum_{i=1}^n{\rm rank}(\hat{W}_i)$,
where

\begin{equation}\label{W_i^}
\hat{W}_i=\begin{pmatrix}
		{\rm Tr}_{F/B}(g_1(\alpha_i)\beta^{(1)})&\cdots &{\rm Tr}_{F/B}(g_1(\alpha_i)\beta^{(m)}) \\
\vdots                            & \ddots & \vdots                                   \\
{\rm Tr}_{F/B}(g_m(\alpha_i)\beta^{(1)})&\cdots & {\rm Tr}_{F/B}(g_m(\alpha_i)\beta^{(m)})\\
\end{pmatrix}.
\end{equation}
Note in (i) and (ii), it has $d=\ell$ which means $\mathcal{A}=\{\alpha_1,...,\alpha_n\}=F$.

(i) If $d=\ell$ and $q>2$, it has $m=1$ from Corollary \ref{condr=2}. Then, $\hat{W}_i=\big({\rm Tr}_{F/B}(g_1(\alpha_i)\beta^{(1)})\big)$. Since $g_1(x)=\eta_1x+\omega_1$ is a bijection from $F$ to $F$ and $\dim_B((\beta^{(1)})^{-1}K)=\ell-1$, we know $$\big|\{i\in[n]: {\rm Tr}_{F/B}(g_1(\alpha_i)\beta^{(1)})\neq0\}\big|=\big|\{i\in[n]: g_1(\alpha_i)\not\in(\beta^{(1)})^{-1}K\}\big|=n-q^{\ell-1}.$$ Then,\begin{align*}
	b&= (n-1)\ell-n+\sum_{i=1}^n{\rm rank}(\hat{W}_i)\\
	   &=(n-1)\ell-n+\big|\{i\in[n]: {\rm Tr}_{F/B}(g_1(\alpha_i)\beta^{(1)})\neq0\}\big|\\
	   &=(n-1)\ell-q^{\ell-1}.
\end{align*}

Next, we prove (ii) and (iii) by counting $\sum_{i=1}^n\big|\{{\bm u}\in B^m:{\bm u}\hat{W}_i={\bm 0}\}\big|$ in two ways. Denote $b_i={\rm rank}(\hat{W}_i)$ for $i\in[n]$. On the one hand, one can obtain that
\begin{equation}\label{new}
\sum_{i=1}^n\big|\{{\bm u}\in B^m:{\bm u}\hat{W}_i={\bm 0}\}\big|=\sum_{i=1}^nq^{m-b_i}.
\end{equation}
On the other hand,
\begin{equation}\label{eq31}
	\sum_{i=1}^n\big|\{{\bm u}\in B^m:{\bm u}\hat{W}_i={\bm 0}\}\big|=\sum_{{\bm u}\in B^m}\big|\{i\in[n]:{\bm u}\hat{W}_i={\bm 0}\}\big|.
\end{equation}	
For each ${\bm u}=(u_1,...,u_m)\in B^m$, ${\bm u}\hat{W}_i=({\rm Tr}_{F/B}(g_{\bm u}(\alpha_i)\beta^{(1)}),..., {\rm Tr}_{F/B}(g_{\bm u}(\alpha_i)\beta^{(m)}))$, where $g_{\bm u}(x)=\sum_{j=1}^mu_jg_j(x)$. Then, ${\bm u}\hat{W}_i={\bm 0}$ if and only if $g_{\bm u}(\alpha_i)\in\bigcap_{j=1}^m(\beta^{(j)})^{-1}K$. If ${\bm u}\neq{\bm 0}$, ${\rm deg}(g_{\bm u})=1$ due to $1)$ of Definition \ref{def}, and therefore $g_{\bm u}(x)$ is a bijection from $F$ to $F$. Thus, $\big|\{i\in[n]:{\bm u}\hat{W}_i={\bm 0}\}\big|=\big|g_{\bm u}(\mathcal{A})\bigcap\big(\bigcap_{j=1}^m(\beta^{(j)})^{-1}K\big)\big|\leq |\bigcap_{j=1}^m(\beta^{(j)})^{-1}K|= q^{\ell-m}$, where the last equality follows from Lemma \ref{lem10}. If ${\bm u}={\bm 0}$, $\big|\{i\in[n]:{\bm u}\hat{W}_i={\bm 0}\}\big|=n=q^{d}$. As a result,
\begin{equation}\label{eq21+}
\sum_{{\bm u}\in B^m}\big|\{i\in[n]:{\bm u}\hat{W}_i={\bm 0}\}\big|=q^d+\!\!\!\sum_{{\bm u}\in B^m\setminus\{{\bm 0}\}}\!\!\!\big|\{i\in[n]:{\bm u}\hat{W}_i={\bm 0}\}\big|\leq q^d+(q^m-1)q^{\ell-m}.
\end{equation}
Note that the repair condition implies $\dim_B\big(\{g_j(0)\}_{j=1}^\ell\big)={\rm rank}(W_1)=\ell$, where $W_1$ has the form as in (\ref{IOmatrix}). Hence ${\rm rank}(W_1)={\rm rank}(\hat{W}_1)+\ell-m$ and $b_1={\rm rank}(\hat{W}_1)=m$. Then, combining \eqref{new}, \eqref{eq31} and \eqref{eq21+}, it has
\begin{equation}\label{eq32+}
\sum_{i=2}^n q^{m-b_i}\leq q^d+q^\ell-q^{\ell-m}-1.
\end{equation}
Let $b_{min}=\min_{0\leq b_i\leq m}\sum_{i=2}^n b_i$ subject to (\ref{eq32+}), then it follows from (\ref{b-formula}) that the repair bandwidth of any I/O optimal repair scheme satisfies
\begin{equation}\label{22}
	b\geq (n-1)(\ell-m)+b_{min}.
\end{equation}	
By the Arithmetic Mean-Geometric Mean inequality, $b_{min}$ is met when  $\{b_i\}_{i=2}^n$ are balanced. However, since $b_2,...,b_n$ can only be integers, we thus estimate $b_{min}$ by choosing integers $b_2,...,b_n$ as small as possible subject to (\ref{eq32+}) and $|b_i-b_j|\leq 1$ for any $i,j\in[2,n]$.

For (ii), it follows from Corollary \ref{condr=2} that $m\in\{1,2\}$. For $m=1$, we have proved in (i) that $b=(q^\ell-1)\ell-q^{\ell-1}$ which is obviously larger than $(2^\ell-1)\ell-3\cdot 2^{\ell-2}$ as $q=2$. For $m=2$,
we first verify that the inequality \eqref{eq32+} holds when $b_2=\cdots=b_n=2$, but does not hold when $b_2=\cdots=b_n=1$. As a result,
we can deduce the minimum of $\sum_{i=2}^nb_i$ arrives as $b_i\in\{1,2\}$ for $i\in[2,n]$. Assume $n_1=|\{i\in[2,n]:b_i=1\}|$ and $n_2=|\{i\in[2,n]:b_i=2\}|$, then $n_1+n_2=n-1=2^\ell-1$. Additionally, we can obtain $2n_1+n_2\leq 2^\ell+3\cdot2^{\ell-2}-1$ from \eqref{eq32+}. Hence
$$\sum^n_{i=1}b_i=2+\sum^n_{i=2}b_i=2+n_1+2n_2=2+3(n_1+n_2)-(2n_1+n_2)\geq 2^{\ell}+2^{\ell-2}.$$ Then, (ii) can be derived because $b=(n-1)\ell-2n+\sum_{i=1}^n b_i$ from (\ref{b-formula}).

For (iii), it follows from Corollary \ref{condr=2} that $m=\ell-d+1$. Similarly, one can verify that the inequality \eqref{eq32+} holds when $b_2=\cdots=b_n=1$, but does not hold when $b_2=\cdots=b_n=0$, so
the minimum of $\sum_{i=2}^nb_i$ arrives as $b_i\in\{0,1\}$ for $i\in[2,n]$. Assume $n_0=|\{i\in[2,n]:b_i=0\}|$ and $n_1=|\{i\in[2,n]:b_i=1\}|$, we can derive
$$
\begin{cases}
	n_0+n_1=q^d-1\\
	q^mn_0+q^{m-1}n_1\leq q^d+q^\ell-q^{\ell-m}-1\\
\end{cases}\;.
$$
Then, $n_1\geq\lceil q^d-q^{2d-\ell-1}-\frac{q^{m}-1}{q^{m-1}(q-1)}\rceil=q^d-q^{2d-\ell-1}-1$, where the last identity follows from  $m=\ell-d+1\geq 2$. Therefore, $\sum^n_{i=1}b_i=m+\sum^n_{i=2}b_i=m+n_1\geq m+q^d-q^{2d-\ell-1}-1$. Finally, (iii) holds due to (\ref{b-formula}).
\end{proof}

\begin{remark}
Theorem \ref{b-r=2-d} indicates that every I/O-optimal linear repair scheme for full-length RS codes with $r=2$ is repair-by-transfer (i.e., the I/O cost equals the repair bandwidth) when $q>2$. When $q=2$, the Construction \uppercase\expandafter{\romannumeral2} in \cite{fullr=2} is an I/O-optimal repair scheme with the repair bandwidth $b=(n-1)\ell-3\cdot 2^{\ell-2}$, which achieves the minimum bandwidth among all I/O-optimal linear repair schemes.
\end{remark}

\subsection{ RS codes with three parities}
In this subsection, we restrict to $B=\mathbb{F}_2$ and $F=\mathbb{F}_{2^\ell}$. In a similar way, we derive lower bounds on the I/O cost and repair bandwidth for RS codes with three parities.
\begin{theorem}\label{r=3} Let $\mathcal{A}$ be a $d$-dimensional $B$-linear subspace of $F$. Then the I/O cost of any linear repair scheme for ${\rm RS}(\mathcal{A},2^d-3)$ over $F$ satisfies:
\begin{equation}\label{IOr3}
		\gamma_{I/O} \geq (n-1)\ell-(\ell-d+2)2^{d-1},
	\end{equation}where $n=2^d$.
\end{theorem}
\begin{proof}
Let $\mathcal{B}=\{\beta^{(1)},...,\beta^{(\ell)}\}$ be a basis of $F$ over $B$. Suppose $\{g_j(x)\}_{j=1}^\ell$ defines an $(m,t)$-normalized repair scheme for node $\alpha_1=0$ with respect to $\mathcal{B}$. Since ${\rm RS}(\mathcal{A},2^d-3)^\bot={\rm RS}(\mathcal{A},3)$, we have ${\rm deg}(g_j)\leq 2$ for $j\in[\ell]$. Moreover, from Definition \ref{def} we may assume
\begin{equation}\label{eq21}
	g_j(x)=\begin{cases}
		\lambda_jx^2+\eta_jx+\omega_j & j\in[m]\\
		\omega_j	& j\in[m+1,\ell]\\
	\end{cases},
\end{equation}
where $\lambda_j,\eta_j,\omega_j\in F$, $\bigcup_{j=m+1}^{\ell}{\rm supp}(\Phi_{\hat{\mathcal{B}}}(\omega_j))=[t+1,\ell]$ and ${\rm dim}(\{\omega_j\}_{j=1}^\ell)=\ell$. Denote $g_{{\bm u}}(x)=\sum_{j=1}^mu_jg_j(x)=\lambda_{{\bm u}}x^2+\eta_{{\bm u}}x+\omega_{{\bm u}}$ for ${\bm u}=(u_1,...,u_m)\in B^m$, where $\lambda_{{\bm u}}=\sum_{j=1}^mu_j\lambda_j$, $\eta_{{\bm u}}=\sum_{j=1}^mu_j\eta_j$, and $\omega_{{\bm u}}=\sum_{j=1}^mu_j\omega_j$.

By Theorem \ref{thmformula}, the I/O cost of $\{g_j(x)\}_{j=1}^\ell$ with respect to $\mathcal{B}$ is
	\begin{align*}\label{eq22}
		\gamma_{I/O}=(n-1)\ell-\frac{1}{2^m}\sum_{s=1}^t\sum_{{\bm u}\in B^m}\sum_{\alpha\in\mathcal{A}}\chi(g_{{\bm u}}(\alpha)\beta^{(s)}),
	\end{align*}
where $\chi$ is the canonical additive character of $F$. First, We calculate
	\begin{align*}
		\sum_{\alpha\in\mathcal{A}}\chi(g_{{\bm u}}(\alpha)\beta^{(s)})&=\chi(\omega_{{\bm u}}\beta^{(s)})\sum_{\alpha\in\mathcal{A}}\chi\big(\beta^{(s)}(\lambda_{{\bm u}}\alpha^2+\eta_{{\bm u}}\alpha)\big).
	\end{align*}
Since $L_{{\bm u}}(x)\triangleq\lambda_{{\bm u}}x^2+\eta_{{\bm u}}x$ is a linearized polynomial over $B$ and $\mathcal{A}$ is a $B$-linear subspace of $F$,  $\beta^{(s)}L_{{\bm u}}(\mathcal{A})$ is also a $B$-linear subspace of $F$. It follows from Lemma \ref{thm4} that
\begin{equation*}
	\sum_{\alpha\in\mathcal{A}}\chi\big(\beta^{(s)}(\lambda_{{\bm u}}\alpha^2+\eta_{{\bm u}}\alpha)\big)=\begin{cases}
			2^d &{\rm if}~ \beta^{(s)}L_{{\bm u}}(\mathcal{A})\subseteq K, \\
			0 &{\rm otherwise}. \end{cases}
\end{equation*}
Furthermore,
denote $U^{(s)}=\{{\bm u}\in B^m:\beta^{(s)}L_{{\bm u}}(\mathcal{A})\subseteq K\}$ and $W^{(s)}=\{\omega_{\bm u}:{\bm u}\in U^{(s)}\}$ for $s\in[t]$, then similar to (\ref{eq26}) and (\ref{eq27+}),  one can derive
\begin{equation}\label{r3}
	\gamma_{I/O}=(n-1)\ell-2^{d-m}\sum_{s=1}^{t'}2^{a_s},
\end{equation}
where $a_s={\rm dim}_B(W^{(s)})={\rm dim}_B(U^{(s)})$,  $t'=|\{s\in[t]:\beta^{(s)}W^{(s)}\subseteq K\}|$ and assume $\{s\in[t]:\beta^{(s)}W^{(s)}\subseteq K\}=[t']$. As in the claim of Theorem \ref{r=2}, it also holds $a_s\leq m-1$ for $s\in[t']$.
We next derive an additional constraint on $\{a_i : i \in [t']\}$. Without loss of generality, we may assume $a_1\geq a_2\geq\cdots\geq a_{t'}$.
	
{\bf Claim.} 	If $t'\geq\ell-d+2$, then $\sum_{i=1}^{\ell-d+2}a_i\leq(\ell-d+1)m$.	

{\it Proof of the claim. }	We first prove $\bigcap_{i=1}^{\ell-d+2}U^{(i)}=\{{\bm 0}\}$. 
By definition of $U^{(i)}$, ${\bm u}\in U^{(i)}$ if and only if $L_{{\bm u}}(\mathcal{A})\subseteq(\beta^{(i)})^{-1}K$. If there exists ${\bm 0}\neq{\bm u}\in B^m$ such that ${\bm u}\in\bigcap_{i=1}^{\ell-d+2}U^{(i)}$, then $L_{{\bm u}}(\mathcal{A})\subseteq\bigcap_{i=1}^{\ell-d+2}(\beta^{(i)})^{-1}K$. On the one hand, since $L_{\bm u}$ is a linearized polynomial over $B$ with degree $\leq2$ and ${\rm dim}_B(\mathcal{A})=d$, it follows that ${\rm dim}_B\big(L_{{\bm u}}(\mathcal{A})\big)\geq d-1$. On the other hand, by Lemma \ref{lem10} it has ${\rm dim}_B\big(\bigcap_{i=1}^{\ell-d+2}(\beta^{(i)})^{-1}K\big)=d-2$. Thus we get a contradiction and so $\bigcap_{i=1}^{\ell-d+2}U^{(i)}=\{{\bm 0}\}$. Then,
according to the dimensionality formula of the sum of linear spaces, it has
\begin{equation}\label{eq27}
	{\rm dim}_B\big(\bigcap_{i=1}^{\ell-d+2}U^{(i)}\big)=\sum_{i=1}^{\ell-d+2}{\rm dim}_B\big(U^{(i)})-\sum^{\ell-d+1}_{j=1}{\rm dim}_B\big(U^{(j+1)}+\bigcap_{i=1}^{j}U^{(i)}\big).
\end{equation}
Since $U^{(i)}\subseteq B^{m}$, we know ${\rm dim}_B\big(U^{(j+1)}+\bigcap_{i=1}^{j}U^{(i)}\big)\leq m$ for $j\in[\ell-d+1]$. Combining with ${\rm dim}_B\big(\bigcap_{i=1}^{\ell-d+2}U^{(i)}\big)=0$ and ${\rm dim}_B(U^{(i)})=a_i$ , we can derive $0\geq \sum_{i=1}^{\ell-d+2}a_i-(\ell-d+1)m$ from (\ref{eq27}). This completes the proof of the claim.

Subsequently, we can deduce that $2^{d-m}\sum_{s=1}^{t'}2^{a_s}\leq(\ell-d+2)2^{d-1}$ by Lemma \ref{r3cond}. Combining with the equality (\ref{r3}), the proof completes.
\end{proof}

\begin{lemma}\label{r3cond}
Let $t',m,d,\ell$ be positive integers and $t'\leq m\leq\ell$. For integers $m-1\geq a_1\geq a_2\geq\cdots\geq a_{t'}\geq 0$ satisfying $\sum_{i=1}^{\min\{t',\ell-d+2\}}a_i\leq (\ell-d+1)m$, it has $2^{d-m}\sum_{i=1}^{t'}2^{a_i}\leq (\ell-d+2)2^{d-1}$, where the equality holds only if $t'=m\leq2(\ell-d+2)$.
\end{lemma}
\begin{proof}
The proof is presented in Appendix \ref{lem17}.
\end{proof}
Based on the proofs of Theorem \ref{r=3} and Lemma \ref{r3cond}, we can conclude the following necessary condition for the lower bound (\ref{IOr3}) to hold with equality.
\begin{corollary}\label{condr=3}
Let $\{g_j(x)\}_{j=1}^\ell$ be a repair scheme of node $0$ for ${\rm RS}(\mathcal{A},2^d-3)$. If $\{g_j(x)\}_{j=1}^\ell$ is  $(m,t)$-normalized with respect to $\mathcal{B}$ and meets the lower bound (\ref{IOr3}) with equality, it must have $t=m\leq2(\ell-d+2)$.
\end{corollary}

Our lower bound (\ref{IOr3}) is tight at $d=\ell$ and $d<\ell$ with $\ell-d+2\mid\ell$. Specifically, when $d=\ell$, our lower bound is $2^{\ell-3}$ higher than the lower bound derived in \cite{I/OFormula}, and is tight due to the Construction 1 given in \cite{I/OFormula}. When $d<\ell$ and $\ell-d+2\mid\ell$, we build an I/O-optimal repair scheme in {\bf Construction 2} in Section \ref{Sec5}. Next, we derive a lower bound on the repair bandwidth for I/O-optimal repair schemes.

\begin{theorem}\label{b-r=3-d}
Suppose $\mathcal{A}$ is a $d$-dimensional $B$-linear subspace of $F$, where $d=\ell$ or $d<\ell$ with $\ell-d+2\mid\ell$. The repair bandwidth, denoted by $b$, of any I/O-optimal linear repair scheme for ${\rm RS}(\mathcal{A},2^d-3)$ satisfies $b\geq(n-1)(d-1)-2^{2d-\ell-1}+\lfloor2^{3d-2\ell-4}\rfloor$.
\end{theorem}

\begin{proof}
Suppose $\{g_j(x)\}_{j=1}^\ell$ defines a repair scheme for node $\alpha_1=0$, which is $(m,t)$-normalized with respect to $\mathcal{B}$ and is I/O-optimal. We may assume $g_j(x)$ is defined as in (\ref{eq21}) for $j\in[\ell]$. By Corollary \ref{condr=3}, it must have $t=m\leq2(\ell-d+2)$. Similar to Theorem \ref{b-r=2-d}, $b$ can be estimated by calculating both sides of (\ref{eq31}). Specifically,
denote $b_i={\rm rank}(\hat{W}_i)$ for $i\in[n]$. The repair condition implies $\dim_B\big(\{g_j(0)\}_{j=1}^\ell\big)={\rm rank}(W_1)=\ell$, thus ${\rm rank}(W_1)={\rm rank}(\hat{W}_1)+\ell-m$ and $b_1={\rm rank}(\hat{W}_1)=m$. It is clear that
\begin{equation}\label{30}
	\sum_{i\in[n]}|\{{\bm u}\in B^m:{\bm u}\hat{W}_i={\bm 0}\}|=\sum_{i\in[n]}2^{m-b_i}=\sum_{i=2}^n2^{m-b_i}+1.
\end{equation}
Next, we compute the right side of (\ref{eq31}). For each ${\bm u}\in B^m$, it can be seen that ${\bm u}\hat{W}_i={\bm 0}$ if and only if $g_{\bm u}(\alpha_i)\in\bigcap_{s=1}^m(\beta^{(s)})^{-1}K$. If ${\bm u}\neq{\bm 0}$, $0<{\rm deg}(g_{\bm u})\leq 2$ due to $1)$ of Definition \ref{def}. Therefore, for any $\gamma\in\mathbb{F}_{2^\ell}$, $|g_{\bm u}^{-1}(\gamma)|\leq 2$. Thus, $|\{i\in[n]:{\bm u}\hat{W}_i={\bm 0}\}|=|\{\alpha\in\mathcal{A}:g_{\bm u}(\alpha)\in\bigcap_{s=1}^m(\beta^{(s)})^{-1}K\}|\leq|\{g_{\bm u}^{-1}(\gamma):\gamma\in\bigcap_{s=1}^m(\beta^{(s)})^{-1}K\}|\leq 2|\bigcap_{s=1}^m(\beta^{(s)})^{-1}K|=2^{\ell-m+1}$, where the last equality follows from Lemma \ref{lem10}. If ${\bm u}={\bm 0}$, $|\{i\in[n]:{\bm u}\hat{W}_i={\bm 0}\}|=n=2^{d}$. Therefore,
\begin{equation}\label{31}
	\sum_{{\bm u}\in B^m}\big|\{i\in[n]:{\bm u}\hat{W}_i={\bm 0}\}\big|=2^d+\!\!\!\sum_{{\bm u}\in B^m\setminus\{{\bm 0}\}}\!\!\!\big|\{i\in[n]:{\bm u}\hat{W}_i={\bm 0}\}\big|\leq 2^d+(2^m-1)2^{\ell-m+1}.
\end{equation}
Combining (\ref{eq31}), (\ref{30}) and (\ref{31}), it has
\begin{equation}\label{28}
	\sum_{i=2}^n 2^{m-b_i}\leq2^{\ell+1}+2^d-2^{\ell-m+1}-1.
\end{equation}
Let $b_{min}^{(m)}=\min_{0\leq b_i\leq m}\sum_{i=2}^n b_i$ subject to (\ref{28}). Recall that $m\in[2(\ell-d+2)]$ and $b_1=m$. It follows from (\ref{b-formula}) that
\begin{equation}\label{29}
	b\geq\min_{m\in[2(\ell-d+2)]} (n-1)(\ell-m)+b_{min}^{(m)}.
\end{equation}	
{\bf Claim.} As $m$ ranges in $[2(\ell-d+2)]$, $(n-1)(\ell-m)+b_{min}^{(m)}$ achieves the minimum at $m=2(\ell-d+2)$.	

{\it Proof of the claim. }
It suffices to show that $$(n-1)(\ell-(s+1))+b_{min}^{(s+1)}\leq (n-1)(\ell-s)+b_{min}^{(s)}$$ for $s\in[2(\ell-d+2)-1]$. Suppose $b_{min}^{(s)}=\sum_{i=2}^n b_i^{(s)}$, where $\{b_i^{(s)}\}_{i=2}^n$ satisfies (\ref{28}) at the case of $m=s$. Let $b_i^{(s+1)}=b_i^{(s)}+1$ for $i\in[2,n]$. It can be seen that

$$\sum_{i=2}^n 2^{s+1-b_i^{(s+1)}}=\sum_{i=2}^n 2^{s-b_i^{(s)}}\leq2^{\ell+1}+2^d-2^{\ell-s+1}-1< 2^{\ell+1}+2^d-2^{\ell-(s+1)+1}-1,$$
where the second inequality is due to $\{b_i^{(s)}\}_{i=2}^n$ satisfying (\ref{28}) at the case of $m=s$. Thus, $\{b_i^{(s+1)}\}_{i=2}^n$ satisfies (\ref{28}) at the case of $m=s+1$, and therefore, $b_{min}^{(s+1)}\leq\sum_{i=2}^n b_i^{(s+1)}$. Then,
\begin{align*}
	(n-1)(\ell-s-1)+b_{min}^{(s+1)}&\leq (n-1)(\ell-s-1)+\sum_{i=2}^n b_i^{(s+1)}\\
	                               &=(n-1)(\ell-s-1)+\sum_{i=2}^n (b_i^{(s)}+1)\\
	                               &=(n-1)(\ell-s)+b_{min}^{(s)}.
\end{align*}
The claim is proved.

To establish a lower bound on the repair bandwidth $b$, we only need to determine $b_{min}^{(m)}$ for $m=2(\ell-d+2)$ according to (\ref{29}) and the claim. Let $m=2(\ell-d+2)$. One can check that $d>\ell-m+1$ and the inequality \eqref{28} holds when $b_2=\cdots=b_n=\ell-d+3$. Moreover, when $d=\ell$ or $d<\ell$ with $\ell-d+2\mid\ell$, it always has $2\leq d\leq\ell$. Then, one can verify that the inequality \eqref{28} does not hold when $b_2=\cdots=b_n=\ell-d+2$, so the minimum of $\sum_{i=2}^nb_i$ arrives as $b_i\in\{\ell-d+2,\ell-d+3\}$ for $i\in[2,n]$. Assume $n_0=|\{i\in[2,n]:b_i=\ell-d+2\}|$ and $n_1=|\{i\in[2,n]:b_i=\ell-d+3\}|$ such that $b_{min}^{(m)}=n_0(\ell-d+2)+n_1(\ell-d+3)$, we know
$$
\begin{cases}
	n_0+n_1=2^d-1\\
	2^{\ell-d+2}n_0+2^{\ell-d+1}n_1\leq 2^{\ell+1}+2^d-2^{\ell-2(\ell-d+2)+1}-1\\
\end{cases}\;.
$$
Then, $n_1\geq\lceil 2^d-2^{2d-\ell-1}+2^{3d-2\ell-4}-2+\frac{1}{2^{\ell-d+1}}\rceil=2^d-2^{2d-\ell-1}-2+\lceil2^{3d-2\ell-4}+\frac{1}{2^{\ell-d+1}}\rceil= 2^d-2^{2d-\ell-1}+\lfloor2^{3d-2\ell-4}\rfloor-1$. Therefore, $$b_{min}^{(m)}=n_0(\ell-d+2)+n_1(\ell-d+3)=(2^d-1)(\ell-d+2)+n_1\geq(2^d-1)(\ell-d+3)-2^{2d-\ell-1}+\lfloor2^{3d-2\ell-4}\rfloor.$$ Combining with (\ref{29}), the theorem is proved.
\end{proof}

\subsection{An I/O optimal Repair scheme with Lower Repair Bandwidth}
From Corollary \ref{condr=3} it implies that any $(m,t)$-normalized repair scheme achieving the optimal I/O cost must satisfy $t=m$. Then according to (\ref{b-formula}), the repair bandwidth of the I/O-optimal repair scheme can be calculated as $b=(n-1)\ell-nm+\sum_{i=1}^n{\rm rank}(\hat{W}_i)$, from which one can see an increase in $m$ may result in a lower bandwidth. Meanwhile, Corollary \ref{condr=3} also gives an upper bound on $m$ for the I/O-optimal scheme, which implies $m\leq 4$ in the case of $d=\ell$. Considering the full-length RS code ${\rm RS}(\mathbb{F}_{2^\ell},2^\ell-3)$, the repair scheme presented in \cite{I/OFormula} achieves the optimal I/O cost with $m=2$, and its repair bandwidth equals the I/O cost. To reduce the bandwidth, we here present an I/O-optimal repair scheme with $m=4$.

\begin{construction}\label{cons1}
Assume $2\mid\ell$ and $\ell\geq 4$. Let $\zeta\in\mathbb{F}_{4}$ be a root of $x^2+x+1$ and $\theta$ be a primitive element of $\mathbb{F}_{2^\ell}$.
\begin{itemize}
  \item[(1)] Select a basis $\mathcal{B}=\{\beta^{(i)}:i\in[\ell]\}$ of $\mathbb{F}_{2^\ell}$ over $\mathbb{F}_{2}$.
	
Particularly, set $\beta^{(1)}=(\theta^2+(\zeta+1)\theta+1)^2,\beta^{(2)}=(\zeta\theta)^2,\beta^{(3)}=1,\beta^{(4)}=(\theta+1)^2$, and then extend $\{\beta^{(1)},\beta^{(2)},\beta^{(3)},\beta^{(4)}\}$ to  $\mathcal{B}=\{\beta^{(1)},...,\beta^{(\ell)}\}$ which forms a basis of $\mathbb{F}_{2^\ell}$ over $\mathbb{F}_2$ \footnote{One can always choose  $\theta$ such that $\beta^{(1)},...,\beta^{(4)}$ are linearly independent over $\mathbb{F}_2$. Note that ${\rm dim}_{\mathbb{F}_2}\big(\{\beta^{(1)},...,\beta^{(4)}\}\big)=4$ if and only if ${\rm dim}_{\mathbb{F}_2}\big(\{\theta^2+(\zeta+1)\theta+1,\zeta\theta,1,\theta+1\}\big)=4$, or equivalently, ${\rm dim}_{\mathbb{F}_2}\big(\{\theta^2,\zeta\theta,1,\theta\}\big)=4$. If $\ell\geq 6$, since $\mathbb{F}_{2^\ell}=\mathbb{F}_{4}(\theta)$ and $[\mathbb{F}_{4}(\theta):\mathbb{F}_4]\geq 3$, then it obviously has ${\rm dim}_{\mathbb{F}_2}\big(\{\theta^2,\zeta\theta,1,\theta\}\big)=4$. If $\ell=4$, let $\theta$ be a root of $x^2+x+\zeta=0$. Then, $\{\theta^2,\zeta\theta,1,\theta\}=\{\theta+\zeta,\zeta\theta,1,\theta\}$, which is a basis of $\mathbb{F}_{2^4}$ over $\mathbb{F}_2$.} .
	
\item[(2)] Construct the repair polynomials $\{g_j(x):j\in[\ell]\}$ for node $0$ of ${\rm RS}(\mathbb{F}_{2^\ell},2^\ell-3)$.
	
Let $\hat{\mathcal{B}}=\{\gamma^{(i)}:i\in[\ell]\}$ be the dual basis of $\mathcal{B}$. Define	
$$ g_j(x)=\begin{cases}
	\lambda_jx^2+\eta_jx+\omega_j &j\in[4]\\
	\gamma^{(j)}	&j\in[5,\ell]	\\
\end{cases},$$
where $\eta_1=1,\eta_2=\zeta,\eta_3=\zeta\theta,\eta_4=\theta$, $\omega_1=\gamma^{(3)},\omega_2=\gamma^{(2)}+\gamma^{(4)},\omega_3=\gamma^{(1)}+\gamma^{(3)},\omega_4=\gamma^{(4)}$, and
\begin{equation}\label{eq29}
	\lambda_j=\begin{cases}
		\eta_j^2\beta^{(1)} &j\in\{1,2\}\\
		\eta_j^2\beta^{(2)}	&j\in\{3,4\}	\\
	\end{cases}.
\end{equation}

\end{itemize}
\end{construction}
By simple computation, one can see that $\lambda_1=\theta^4+\zeta\theta^2+1,\lambda_2=(\zeta+1)\theta^4+\theta^2+\zeta^2,
\lambda_3=\zeta\theta^4,\lambda_4=\zeta^2\theta^4$. Moreover, it can be verified that
\begin{equation}\label{eq34}
\begin{cases}
  \lambda_1+\lambda_3+\lambda_4=\beta^{(3)}(\eta_1+\eta_3+\eta_4)^2\\
  \lambda_2+\lambda_4=\beta^{(3)}(\eta_2+\eta_4)^2\\
  \lambda_2+\lambda_3+\lambda_4=\beta^{(4)}(\eta_2+\eta_3+\eta_4)^2\\
  \lambda_1+\lambda_3=\beta^{(4)}(\eta_1+\eta_3)^2
\end{cases}.
\end{equation}
These relations will be used in later proofs.

Noting that $\{g_j(0)\}_{j=1}^\ell=\{\gamma^{(3)},\gamma^{(2)}+\gamma^{(4)},\gamma^{(1)}+\gamma^{(3)},\gamma^{(4)},\gamma^{(5)},...,\gamma^{(\ell)}\}$, the repair condition is obviously satisfied since $\{\gamma^{(1)},...,\gamma^{(\ell)}\}$ is a basis of $\mathbb{F}_{2^{\ell}}$ over $\mathbb{F}_2$. Hence $\{g_j(x)\}_{j=1}^\ell$ defines a repair scheme for node $\alpha_1=0$ in ${\rm RS}(\mathbb{F}_{2^\ell},2^\ell-3)$. Next, we estimate the I/O cost and the repair bandwidth of this repair scheme.
	
It can be seen that the coefficients of $x$ in $\{g_j(x)\}_{j=1}^4$ are $\{1,\zeta,\zeta\theta,\theta\}$, which are linearly independent over $\mathbb{F}_2$. Thus, ${\rm deg}(g_{\bm u})>0$ for ${\bm u}\in \mathbb{F}_2^4\setminus\{{\bm 0}\}$. For $j\in[5,\ell]$, $g_j(x)=\gamma^{(j)}$ and ${\rm supp}(\Phi_{\hat{\mathcal{B}}}(\gamma^{(j)}))=\{j\}$, hence $\bigcup_{j=5}^{\ell}{\rm supp}(\Phi_{\hat{\mathcal{B}}}(\omega_j))=[5,\ell]$. Therefore, $\{g_j(x)\}_{j=1}^{\ell}$ is $(4,4)$-normalized with respect to $\mathcal{B}$. Then according to (\ref{r3}), the I/O cost of $\{g_j(x)\}_{j=1}^\ell$ with respect to $\mathcal{B}$ is
\begin{equation}\label{35}
(n-1)\ell-2^{\ell-4}\sum_{s=1}^{t'}2^{a_s},
\end{equation}		
where $a_s={\rm dim}_B(U^{(s)})$, $t'=|\{s\in[t]:\beta^{(s)}W^{(s)}\subseteq K\}|$, and $U^{(s)},W^{(s)}, K$ are defined as in Theorem \ref{r=3}. Since $W^{(s)}=\{\omega_{\bm u}:{\bm u}\in U^{(s)}\}$ is fully determined by $U^{(s)}$, determining $U^{(s)}$ is sufficient for calculating the I/O cost. The following claim helps to characterize the $U^{(s)}$.
	
{\bf Claim.} For $s\in[4]$, $U^{(s)}=\{{\bm u}\in \mathbb{F}_2^4:\lambda_{{\bm u}}=\beta^{(s)}\eta_{{\bm u}}^2\}$.
	
{\it proof of the claim.}
Recall the definition of $U^{(s)}$, i.e., $U^{(s)}=\{{\bm u}\in \mathbb{F}_2^4:\beta^{(s)}L_{{\bm u}}(\mathbb{F}_{2^\ell})\subseteq K\}$, where $L_{{\bm u}}(x)=\lambda_{{\bm u}}x^2+\eta_{{\bm u}}x$. It implies $U^{(s)}=\{{\bm u}\in \mathbb{F}_2^4:\forall\alpha\in\mathbb{F}_{2^\ell}, {\rm Tr}(\beta^{(s)}\lambda_{{\bm u}}\alpha^2)={\rm Tr}(\beta^{(s)}\eta_{{\bm u}}\alpha)\}$. Given that ${\rm Tr}(\beta^{(s)}\eta_{{\bm u}}\alpha)={\rm Tr}\big((\beta^{(s)}\eta_{{\bm u}}\alpha)^2\big)$, we can deduce $$U^{(s)}=\{{\bm u}\in \mathbb{F}_2^4:\forall\alpha\in\mathbb{F}_{2^\ell}, {\rm Tr}(\beta^{(s)}\lambda_{{\bm u}}\alpha^2)={\rm Tr}\big((\beta^{(s)}\eta_{{\bm u}}\alpha)^2\big)\}=\{{\bm u}\in \mathbb{F}_2^4:\forall\alpha\in\mathbb{F}_{2^\ell}, {\rm Tr}\big((\lambda_{{\bm u}}-\beta^{(s)}\eta_{{\bm u}}^2)\beta^{(s)}\alpha^2\big)=0\}.$$ Since $\{\beta^{(s)}\alpha^2:\alpha\in\mathbb{F}_{2^\ell}\}=\mathbb{F}_{2^\ell}$, it follows $U^{(s)}=\{{\bm u}\in\mathbb{F}_2^4:\lambda_{{\bm u}}=\beta^{(s)}\eta_{{\bm u}}^2\}$ and the claim is proved.
	
\begin{proposition}\label{prop}
The I/O cost with respect to $\mathcal{B}$ of the repair scheme given in Construction \ref{cons1} for ${\rm RS}(\mathbb{F}_{2^\ell},2^\ell-3)$ is optimal according to Theorem \ref{r=3}, i.e.,
\begin{equation*}
	\gamma_{I/O}=(n-1)\ell-2^\ell,
\end{equation*}
where $n=2^\ell$. Moreover, the repair bandwidth $b$ of this scheme satisfies $(n-1)(\ell-1)-2^{\ell-1}+2^{\ell-4}\leq b\leq (n-1)\ell-2^\ell$.
\end{proposition}

\begin{proof}
We will show that $t'=4$ and $a_s=2$ for $s\in[4]$, where $t'$ and $a_s$ are defined as in Theorem \ref{r=3}. Then, the I/O cost of $\{g_j(x)\}_{j=1}^\ell$ with respect to $\mathcal{B}$ is $(n-1)\ell-2^\ell$ according to (\ref{35}), which is optimal. Additionally, the bounds on the repair bandwidth $b$ follow from Theorem \ref{b-r=3-d} and the fact that $b\leq \gamma_{I/O}$.
		
We first show that $a_s=2$, i.e., ${\rm dim}_{\mathbb{F}_2}(U^{(s)})=2$ for $s\in[4]$. Based on the above claim, we know $U^{(s)}=\{{\bm u}\in \mathbb{F}_2^4:\lambda_{{\bm u}}=\beta^{(s)}\eta_{{\bm u}}^2\}$, $s\in[4]$. Denote ${\bm e}_i=(0,...,\stackrel{i}{1},...,0)\in \mathbb{F}_2^4$. From (\ref{eq29}), it directly follows that ${\bm e}_1,{\bm e}_2\in U^{(1)}$ and ${\bm e}_3,{\bm e}_4\in U^{(2)}$. Also, from (\ref{eq34}) one can see that ${\bm e}_1+{\bm e}_3+{\bm e}_4,{\bm e}_2+{\bm e}_4\in U^{(3)}$ and ${\bm e}_2+{\bm e}_3+{\bm e}_4,{\bm e}_1+{\bm e}_3\in U^{(4)}$. Therefore, ${\rm dim}_{{\mathbb{F}_2}}(U^{(s)})\geq2$ for $s\in[4]$. Since $\beta^{(1)},...,\beta^{(4)}$ are linearly independent, we know $U^{(1)},U^{(2)},U^{(3)},U^{(4)}$ are pairwise disjoint. Hence, $4\leq{\rm dim}_{{\mathbb{F}_2}}(U^{(i)})+{\rm dim}_{{\mathbb{F}_2}}(U^{(j)})={\rm dim}_{{\mathbb{F}_2}}(U^{(i)}+U^{(j)})\leq 4$ for different $i,j\in[4]$, which implies that ${\rm dim}_{{\mathbb{F}_2}}(U^{(s)})=2$ for $s\in[4]$ and
\begin{equation}\label{Us}\begin{cases}
	U^{(1)}={\rm span}_{\mathbb{F}_2}\big(\{{\bm e}_1,{\bm e}_2\}\big)\\
	U^{(2)}={\rm span}_{\mathbb{F}_2}\big(\{{\bm e}_3,{\bm e}_4\}\big)\\
	U^{(3)}={\rm span}_{\mathbb{F}_2}\big(\{{\bm e}_1+{\bm e}_3+{\bm e}_4,{\bm e}_2+{\bm e}_4\}\big)\\
	U^{(4)}={\rm span}_{\mathbb{F}_2}\big(\{{\bm e}_2+{\bm e}_3+{\bm e}_4,{\bm e}_1+{\bm e}_3\}\big)
\end{cases}\;.\end{equation}
		
We now show that $t'=4$, i.e., $\beta^{(s)}W^{(s)}\subseteq K$ for $s\in[4]$. Specifically, $W^{(s)}=\{\omega_{\bm u}:{\bm u}\in U^{(s)}\}$ which can be easily computed from $U^{(s)}$, i.e.,
$$\begin{cases}
	W^{(1)}={\rm span}_{\mathbb{F}_2}\big(\{\omega_{1},\omega_2\}\big)={\rm span}_{\mathbb{F}_2}\big(\{\gamma^{(3)},\gamma^{(2)}+\gamma^{(4)}\}\big)\\
	W^{(2)}={\rm span}_{\mathbb{F}_2}\big(\{\omega_{3},\omega_4\}\big)={\rm span}_{\mathbb{F}_2}\big(\{\gamma^{(1)}+\gamma^{(3)},\gamma^{(4)}\}\big)\\
	W^{(3)}={\rm span}_{\mathbb{F}_2}\big(\{\omega_{1}+\omega_{3}+\omega_{4},\omega_2+\omega_{4}\}\big)={\rm span}_{\mathbb{F}_2}\big(\{\gamma^{(1)}+\gamma^{(4)},\gamma^{(2)}\}\big)\\
	W^{(4)}={\rm span}_{\mathbb{F}_2}\big(\{\omega_{2}+\omega_{3}+\omega_{4},\omega_1+\omega_{3}\}\big)={\rm span}_{\mathbb{F}_2}\big(\{\gamma^{(1)}+\gamma^{(2)}+\gamma^{(3)},\gamma^{(1)}\}\big)
\end{cases}\;.$$
Since $\{\gamma^{(1)},...,\gamma^{(\ell)}\}$ is the dual basis of $\{\beta^{(1)},...,\beta^{(\ell)}\}$, it obviously holds that $\beta^{(s)}W^{(s)}\subseteq K$, $s\in[4]$.
\end{proof}

Although we cannot prove that Construction \ref{cons1} always has a repair bandwidth lower than the I/O cost, it is likely to occur with a proper choice of the primitive element $\theta$. According to (\ref{b-formula}), the repair bandwidth depends on the exact value of $\sum_{i=1}^n{\rm rank}(\hat{W}_i)$. Combining with (\ref{7}), if there exists some $i\in[n]$ such that ${\rm rank}(\hat{W}_i)<{\rm nz}(\hat{W}_i)$, then the repair bandwidth is lower than the I/O cost. We provide such an example in the following.

\begin{example}
Set $\ell=4$ and let $\theta$ be a root of $x^2+x+\zeta$ in Construction \ref{cons1}.
Then, it can be checked that the basis $\mathcal{B}=\{\beta^{(1)}=\theta^{14},\beta^{(2)}=\theta^{12},\beta^{(3)}=1,\beta^{(4)}=\theta^{8}\}$ and the dual basis $\hat{\mathcal{B}}=\{\gamma^{(1)}=\theta^{8},\gamma^{(2)}=\theta^{2},\gamma^{(3)}=\theta^{11},\gamma^{(4)}=\theta^{5}\}$.
According to Construction \ref{cons1}, it has
\begin{equation*}
	\begin{cases}
		g_1(x)=\theta^{14}x^2+x+\gamma^{(3)}\\
		g_2(x)=\theta^{9}x^2+\theta^{5}x+\gamma^{(2)}+\gamma^{(4)}\\
		g_3(x)=\theta^{9}x^2+\theta^{6}x+\gamma^{(1)}+\gamma^{(3)}\\
		g_4(x)=\theta^{14}x^2+\theta x+\gamma^{(4)}
	\end{cases}\;.
\end{equation*}
Denote $\mathbb{F}_{2^4}=\{\alpha_1=0,\alpha_2...,\alpha_{16}\}$ where $\alpha_i=\theta^{i-1}$ for $i\in[2,16]$.
Recall that  $\hat{W}_i$ is defined as $\hat{W}_i=\big({\rm Tr}(g_{\mu}(\alpha_i)\beta^{(\nu)})\big)_{\mu,\nu\in[4]}$. By a detailed computation, one can see
that
\begin{equation*}
	\hat{W}_7=\begin{pmatrix}
			0 & 1 & 1 & 1\\
			0 & 0 & 0 & 0\\	
			0& 0 & 1 & 1\\	
			1& 0 & 0& 1\\	
	\end{pmatrix},	\hat{W}_9=\begin{pmatrix}
	0 & 1 & 1 & 0\\
	0 & 0 & 1 & 1\\	
	1& 0 & 0 & 0\\	
	1& 0 & 1& 1\\	
	\end{pmatrix}, 	\hat{W}_{11}=\begin{pmatrix}
	0 & 1 & 1 & 0\\
	0 & 1 & 1 & 0\\	
	0& 0 & 0 & 0\\	
	1& 0 & 1& 0\\	
	\end{pmatrix},	\end{equation*}
are all the $\hat{W}_i$'s satisfying ${\rm rank}(\hat{W}_i)<{\rm nz}(\hat{W}_i)$. Therefore, the repair scheme in this case has a repair bandwidth of $41$, which is less than the I/O cost of $44$.
\end{example}

Furthermore, we examine the repair bandwidth and I/O cost of our schemes with additional examples by running a program using inherent primitive elements. The results show that, although the bandwidth of our scheme does not reach the lower bound established in Theorem \ref{b-r=3-d}, it outperforms existing schemes. In Table \ref{table3}, we present a comparison of the repair bandwidth and I/O cost between Construction \ref{cons1} and previous schemes for ${\rm RS}(\mathbb{F}_{2^\ell},2^\ell-3)$ at $\ell=2e$ with $e$ ranging from 2 to 7.

\begin{table}[H]
\caption{\scriptsize Comparison of repair bandwidth and I/O cost (in bit) of the repair schemes for ${\rm RS}(\mathbb{F}_{2^\ell},2^\ell-3)$.}\label{table3}
\setlength{\belowcaptionskip}{6pt}
	\begin{subtable}{.5\linewidth}
		\renewcommand{\arraystretch}{1.2}
		\centering
		\setlength{\tabcolsep}{1mm}
		{        \begin{tabular}{|l|c|c|c|c|c|c|}
				\hline
				~~~~~~~~~~~~$n$ & $2^4$ & $2^6$ & $2^8$ &  $2^{10}$ & $2^{12}$ &$2^{14}$\\
				\hline
				Schemes in \cite{RSrepair,obRS} & 45 & 315& 1785 & 9207 & 45045 & 212979 \\
				\hline
				Scheme in \cite{I/OFormula} & 44 & 314 & 1784 & 9206 &45044 & 212978\\
				\hline
				Construction \ref{cons1} & \textbf{41} & \textbf{300} & \textbf{1733} & \textbf{9002} &\textbf{44228} & \textbf{209714}\\
				\hline
		\end{tabular}}
	\caption{Comparison of the repair bandwidth}
    \end{subtable}%
\setlength{\belowcaptionskip}{6pt}
    \begin{subtable}{.5\linewidth}
		\renewcommand{\arraystretch}{1.2}
		\centering
		\setlength{\tabcolsep}{0.8mm}
		{\begin{tabular}{|l|c|c|c|c|c|c|}
				\hline
				~~~~~~~~~~~~$n$ & $2^4$ & $2^6$ & $2^8$ &  $2^{10}$ & $2^{12}$ &$2^{14}$\\
				\hline
				Schemes in \cite{RSrepair,obRS} & 56 & 372& 2032 & 10220 & 49128 & 229348 \\
				\hline
				Scheme in \cite{I/OFormula} & 44 & 314 & 1784 & 9206 &45044 & 212978\\
				\hline
				Construction \ref{cons1} & 44 & 314 & 1784 & 9206 &45044 & 212978\\
				\hline
		\end{tabular}}
	\caption{Comparison of the I/O cost}
    \end{subtable}%
\end{table}

\section{Linear repair scheme for RS codes evaluated on a $B$-linear subspace}\label{Sec5}
In this section, we construct a family of RS codes ${\rm RS}(\mathcal{A},k)$ that have repair schemes with reduced I/O cost, where $\mathcal{A}$ is a $d$-dimensional $B$-linear subspace of $F$.  When $n-k=2$, with $\ell-d+1\mid\ell$, or when $n-k=3$, with $\ell-d+2\mid\ell$, the I/O cost of our scheme matches the lower bounds established in Theorem \ref{r=2} and Theorem \ref{r=3}, respectively.

Before presenting the construction, we recall a lemma from \cite{I/OFormula}, which is used to select normalized polynomials.
\begin{lemma}[\cite{I/OFormula}]\label{q-poly}
Assume $1\leq t<\ell$ and $\beta_1,...,\beta_t\in F$ are linearly independent over $B$. Let $(\theta_0,\theta_1,...,\theta_t)$ with $\theta_t\neq 0$ be a solution to the system
\begin{equation*}
	\left( {\begin{array}{cccc} \beta_1 & \beta_1^{q} &  \cdots &\beta_1^{q^t}\\
			\beta_2 & \beta_2^{q} &  \cdots &\beta_2^{q^t}\\
			\vdots & \vdots & \ddots & \vdots \\
			\beta_t & \beta_t^{q} &  \cdots &\beta_t^{q^t}\end {array}}\right) \left({\begin{array}{c}
				\theta_t  \\
				\theta_{t-1}^q\\
				\vdots\\
				\theta_0^{q^t}
		\end{array}}\right)={\bm 0}.
\end{equation*}
Define $L(x)=\sum_{j=0}^t \theta_j x^{q^{j}}$. Then $L(F)=\bigcap_{i=1}^t \beta_i^{-1}K$.
\end{lemma}
	
\begin{construction}\label{cons2} Let $r=q^d-k\geq q^s+1$ for some $s\geq 0$.
Assume $m\leq\ell-d+s+1$ and $m\mid\ell$. 
\begin{itemize}
\item[(1)] Select a basis $\hat{\mathcal{B}}=\{\gamma^{(1)},...,\gamma^{(\ell)}\}$ of $F$ over $B$.

Noting that $B\subseteq \mathbb{F}_{q^m}\subseteq F$,		
let $\{\gamma^{(1)}\!=\!1,\gamma^{(2)},...,\gamma^{(m)}\}$ be a basis of $\mathbb{F}_{q^m}$ over $B$ and $\{\lambda_1=1,\lambda_2,...,\lambda_{\frac{\ell}{m}}\}$ be a basis of $F$ over $\mathbb{F}_{q^m}$. Consequently, $\{\lambda_i\gamma^{(j)}:i\in[\frac{\ell}{m}],j\in[m]\}$ forms a basis of $F$ over $B$. Set $\gamma^{\left((i-1)m+j\right)}=\lambda_i\gamma^{(j)}$ for $i\in[\frac{\ell}{m}]$ and $j\in[m]$.

\item[(2)] Define the set of evaluation points $\mathcal{A}$ such that ${\rm dim}_B(\mathcal{A})=d<\ell$\;.
		
Let $\mathcal{B}=\{\beta^{(i)}:i\in[\ell]\}$ be the dual basis of $\hat{\mathcal{B}}$. For $s=0$, set $L(x)=\alpha x$, where $\alpha\neq 0$ is arbitrarily chosen from $F$.
For $s>0$, construct a $q$-polynomial $L(x)\in F[x]$ of degree $q^s$ such that $L(F)=\bigcap_{i=2}^{s+1}(\beta^{(i)})^{-1}K$ by Lemma \ref{q-poly}. Denote $W=\bigcap_{i=2}^{\ell-d+s+1}(\beta^{(i)})^{-1}K\subseteq L(F)$. Let $\mathcal{A}=L^{-1}(W)$. According to Lemma \ref{lem10}, we have ${\rm dim}_B(L(F))=\ell-s$ and ${\rm dim}_B(W)=d-s$. Consequently, ${\rm dim}_B({\rm Ker}(L))=s$, and thus ${\rm dim}_B(\mathcal{A})={\rm dim}_B(L^{-1}(W))={\rm dim}_B(W)+s=d$.

\item[(3)] Construct the repair polynomials $\{g_j(x):j\in[\ell]\}$ for node $0$ in ${\rm RS}(\mathcal{A},n-r)$.
		
Define
$$ g_j(x)=\begin{cases}
		\gamma^{(j)}L(x)+\gamma^{(j)} &j\in[m]\\
		\gamma^{(j)}	&j\in[m+1,\ell]	\\
	\end{cases}.$$
\end{itemize}
\end{construction}

Since ${\rm deg}(g_j(x))\leq q^s< r$ and $\{g_j(0)\}_{j=1}^\ell$= $\hat{\mathcal{B}}$, $\{g_{j}(x)\}_{j=1}^\ell$ obviously defines a linear repair scheme of node $0$. Next, we compute the I/O cost and repair bandwidth of the repair scheme.

\begin{theorem}\label{I/O-cons2}
The linear repair scheme for ${\rm RS}(\mathcal{A},q^d-r)$ defined in Construction \ref{cons2} incurs an I/O cost $\gamma_{I/O}=(n-1)\ell-mq^{d-1}$ with respect to $\mathcal{B}$. Moreover, the repair bandwidth of the scheme is equal to its I/O cost.
\end{theorem}
\begin{proof}
Denote $g_{\bm u}(x)=\sum_{j=1}^m u_j g_j(x)=\sum_{j=1}^m u_j\gamma^{(j)}L(x)+\sum_{j=1}^m u_j\gamma^{(j)}$ for ${\bm u}=(u_1,...,u_m)\in B^m$. Since $\{\gamma^{(1)},...,\gamma^{(m)}\}$ are linearly independent over $B$, it follows that ${\rm deg}(g_{\bm u}(x))>0$ for ${\bm u}\in B^m\setminus\{{\bm 0}\}$. Moreover, it can be seen that $\bigcup_{j=m+1}^{\ell}{\rm supp}(\Phi_{\hat{\mathcal{B}}}(g_j(x))=\bigcup_{j=m+1}^{\ell}{\rm supp}(\Phi_{\hat{\mathcal{B}}}(\gamma^{(j)}))=[m+1,\ell]$. Thus, $\{g_j(x)\}_{j=1}^\ell$ is $(m,m)$-normalized with respect to $\mathcal{B}$. It follows from (\ref{7}) and (\ref{b-formula}) that the I/O cost and repair bandwidth of the scheme depend on the calculation of ${\rm nz}(\hat{W}_i)$ and ${\rm rank}(\hat{W}_i)$, respectively, where $\hat{W}_i=\big({\rm Tr}(g_{\mu}(\alpha_i)\beta^{(\nu)})\big)_{\mu,\nu\in[m]}$ and $\mathcal{A}=\{\alpha_1=0,\alpha_2,...,\alpha_n\}$.

Next, we characterize $g_1(\mathcal{A}),...,g_m(\mathcal{A})$ to determine $\hat{W}_i$. Recall that $L(\mathcal{A})=W=\bigcap_{i=2}^{\ell-d+s+1}(\beta^{(i)})^{-1}K$. Since $\{\gamma^{(1)},...,\gamma^{(\ell)}\}$ is the dual basis of $\{\beta^{(1)},...,\beta^{(\ell)}\}$, it actually holds $L(\mathcal{A})={\rm span}_B\big(\{\gamma^{(1)},\gamma^{(\ell-d+s+2)},...,\gamma^{(\ell)}\}\big)$. Moreover, we have the following claim.
			
{\bf Claim.} For $j\in[m]$, $\gamma^{(j)}\in \gamma^{(j)}L(\mathcal{A})$ and $\gamma^{(j)}L(\mathcal{A})\subseteq{\rm span}_B\big(\{\gamma^{(j)},\gamma^{(m+1)},...,\gamma^{(\ell)}\}\big)$.
			
{\it Proof of the claim. } Since $\gamma^{(j)}L(\mathcal{A})={\rm span}_B\big(\{\gamma^{(j)}\gamma^{(1)},\gamma^{(j)}\gamma^{(\ell-d+s+2)},...,\gamma^{(j)}\gamma^{(\ell)}\}\big)$ and $\gamma^{(1)}=1$, it suffices to show ${\rm span}_B\big(\{\gamma^{(j)}\gamma^{(\ell-d+s+2)},...,\gamma^{(j)}\gamma^{(\ell)}\}\big)\subseteq{\rm span}_B\big(\{\gamma^{(m+1)},...,\gamma^{(\ell)}\}\big)$. By the definition of $\{\gamma^{(1)},...,\gamma^{(\ell)}\}$, we have ${\rm span}_B\big(\{\gamma^{(1)},...,\gamma^{(m)}\}\big)=\mathbb{F}_{q^m}$ and ${\rm span}_B\big(\{\gamma^{(m+1)},...,\gamma^{(\ell)}\}\big)=
\lambda_2\mathbb{F}_{q^m}\oplus\lambda_3\mathbb{F}_{q^m}\oplus\cdots\oplus\lambda_{\frac{\ell}{m}}\mathbb{F}_{q^m}$.
As a result, for $j\in[m]$, it has
\begin{align*}
& {\rm span}_B\big(\{\gamma^{(j)}\gamma^{(m+1)},...,\gamma^{(j)}\gamma^{(\ell)}\}\big)\\=&
\gamma^{(j)}\lambda_2\mathbb{F}_{q^m}\oplus\gamma^{(j)}\lambda_3\mathbb{F}_{q^m}\oplus\cdots\oplus\gamma^{(j)}\lambda_{\frac{\ell}{m}}\mathbb{F}_{q^m}\\
=&\lambda_2\mathbb{F}_{q^m}\oplus\lambda_3\mathbb{F}_{q^m}\oplus\cdots\oplus\lambda_{\frac{\ell}{m}}\mathbb{F}_{q^m}\\
=&{\rm span}_B\big(\{\gamma^{(m+1)},...,\gamma^{(\ell)}\}\big)
\end{align*}
where the second equality comes from $\gamma^{(j)}\mathbb{F}_{q^m}=\mathbb{F}_{q^m}$ because $\gamma^{(j)}\in\mathbb{F}_{q^m}$ for $j\in[m]$.  Finally, since $m\leq\ell-d+s+1$, it obviously has
${\rm span}_B\big(\{\gamma^{(j)}\gamma^{(\ell-d+s+2)},...,\gamma^{(j)}\gamma^{(\ell)}\}\big)\subseteq{\rm span}_B\big(\{\gamma^{(j)}\gamma^{(m+1)},...,\gamma^{(j)}\gamma^{(\ell)}\}\big)$ and the claim is proved.

By the definition of $g_{\mu}(x)$, for $\mu,\nu\in[m]$, it has
\begin{align}\label{eq43}
		{\rm Tr}_{F/B}(g_{\mu}(\alpha_i)\beta^{(\nu)})&={\rm Tr}_{F/B}(\gamma^{(\mu)}L(\alpha_i)\beta^{(\nu)})+{\rm Tr}_{F/B}(\gamma^{(\mu)}\beta^{(\nu)})\notag\\
				&=\begin{cases}
				{\rm Tr}_{F/B}(\gamma^{(\mu)}L(\alpha_i)\beta^{(\mu)})+1, &{\text{ if}~}\nu=\mu\\
				0, &\text{ otherwise}~
			\end{cases}
\end{align}
where (\ref{eq43}) follows from the claim and the fact ${\rm Tr}(\beta^{(i)}\gamma^{(j)})\!=\!{\bm 1}_{i=j}$. Therefore, $\hat{W}_i$, $i\in[n]$, has the following form
\begin{equation*}
	\hat{W}_i=\left( {\begin{array}{cccc}
				{\rm Tr}_{F/B}(\gamma^{(1)}L(\alpha_i)\beta^{(1)})+1 & 0 & \cdots & 0 \\
				0&{\rm Tr}_{F/B}(\gamma^{(2)}L(\alpha_i)\beta^{(2)})+1 &\cdots &0\\
				\vdots&\vdots &\ddots  & \vdots \\
				0 &0&\cdots &{\rm Tr}_{F/B}(\gamma^{(m)}L(\alpha_i)\beta^{(m)})+1\\
				\end {array}}\right),
\end{equation*}
which is an $m\times m$ diagonal matrix. It immediately follows ${\rm rank}(\hat{W}_i)={\rm nz}(\hat{W}_i)$, so the repair bandwidth and I/O cost are equal. Moreover,
\begin{align}\notag
\sum_{i\in[n]}{\rm nz}(\hat{W}_i)&=\sum_{i\in[n]}\big(m-\big|\{j\in[m]:{\rm Tr}_{F/B}(\gamma^{(j)}L(\alpha_i)\beta^{(j)})+1=0\}\big|\big)\\
			                     &=nm-\sum_{j\in[m]}\big|\{i\in[n]:{\rm Tr}_{F/B}(\gamma^{(j)}L(\alpha_i)\beta^{(j)})=-1\}\big|\label{eq322}\\
					             &=nm-\sum_{j\in[m]}q^s\big|\{\alpha\in \gamma^{(j)}L(\mathcal{A}):{\rm Tr}_{F/B}(\alpha\beta^{(j)})=-1\}\big|, \label{eq33}
\end{align}
where (\ref{eq322}) uses double counting and (\ref{eq33}) comes from $\dim_B({\rm Ker}(L))=s$. Recall that ${\rm dim}_B(L(\mathcal{A}))=d-s$.  For $j\in[m]$, combining with the claim, we may assume $\gamma^{(j)}L(\mathcal{A})={\rm span}_B\{\gamma^{(j)}\}\oplus V_j$, where $V_j\subseteq{\rm span}_B\big(\{\gamma^{(m+1)},...,\gamma^{(\ell)}\}\big)$ and ${\rm dim}_B(V_j)=d-s-1$. Since ${\rm span}_B\big(\{\gamma^{(m+1)},...,\gamma^{(\ell)}\}\big)\subseteq (\beta^{(j)})^{-1}K$ for $j\in[m]$, we can deduce $\{\alpha\in \gamma^{(j)}L(\mathcal{A}):{\rm Tr}_{F/B}(\alpha\beta^{(j)})=-1\}=-\gamma^{(j)}+V_j$. As a result, $\big|\{\alpha\in \gamma^{(j)}L(\mathcal{A}):{\rm Tr}_{F/B}(\alpha\beta^{(j)})=-1\}\big|=|V_j|=q^{d-s-1}$. Combining with (\ref{eq9+}) and (\ref{eq33}), we obtain $\gamma_{I/O}=(n-1)\ell-mq^{d-1}$.
\end{proof}	
\begin{remark}
Construction \ref{cons2} is inspired by Construction 1 in \cite{I/OFormula}. Both constructions control the  I/O cost by designing the image set of the repair polynomials. Specifically, for RS code ${\rm RS}(\mathcal{A},q^d-r)$ over $F$ with $r\geq q^s+1$, where $\mathcal{A}$ is a $d$-dimensional subspace of $F$, the crux of the construction is to identify $q$-polynomials $L_1(x),...,L_m(x)$ of degree $q^s$ such that $L_i(\mathcal{A})=V_i$ for $i\in[m]$, where $V_i$ is a $(d-s)$-dimensional subspace and lies in $\bigcap_{j\in[m],j\neq i}(\beta^{(j)})^{-1}K$. Construction 1 in \cite{I/OFormula} is for full-length RS codes, i.e., $\mathcal{A}=F$. Thus these $q$-polynomials can be easily found using Lemma \ref{q-poly}. However, it is not easy for the case of $d<\ell$. In Construction \ref{cons2} of this work, we select a $q$-polynomial $L(x)$ and some $\alpha_i\in F$  such that $L(\mathcal{A})=V$ and $\alpha_iV=V_i$ for $i\in[m]$. Then, defining $L_i(x)=\alpha_iL(x)$ ensures that $L_i(\mathcal{A})=V_i$. Based on this idea, we successfully find the desired repair polynomials, but it also brings the restriction that $m\mid\ell$.
\end{remark}

In \cite{shortr=2}, the authors build some repair schemes for RS codes ${\rm RS}(\mathcal{A},2^d-2)$, where $\mathcal{A}$ is a $d$-dimensional subspace of $\mathbb{F}_{2^\ell}$. However, the I/O cost of their schemes does not achieve the lower bound established in Theorem \ref{r=2}. In contrast, our Construction \ref{cons2} attains the lower bound in Theorem \ref{r=2} when $\ell-d+1 \mid \ell$.
\begin{corollary}\label{coro28}
	Set $s=0$ and $r=2$. Suppose $m=\ell-d+1\mid\ell$, Construction \ref{cons2} provides an RS code ${\rm RS}(\mathcal{A},q^d-2)$ and a corresponding linear repair scheme. The I/O cost of that linear repair scheme is $(n-1)\ell-(\ell-d+1)q^{d-1}$, which is optimal according to Theorem \ref{r=2}.
\end{corollary}
Moreover, for RS codes ${\rm RS}(\mathcal{A},2^d-3)$, our scheme achieves optimal I/O cost when $\ell-d+2\mid\ell$.
\begin{corollary}\label{coro29}
	Set $s=1$ and $r=2$. Suppose $m=\ell-d+2\mid\ell$, Construction \ref{cons2} provides an RS code ${\rm RS}(\mathcal{A},2^d-3)$ and a corresponding linear repair scheme. The I/O cost of that linear repair scheme is $(n-1)\ell-(\ell-d+2)2^{d-1}$, which is optimal according to Theorem \ref{r=3}.
\end{corollary}
			
In general, the difference between the trivial I/O cost and our scheme's I/O cost is $(n-r)\ell-\big((n-1)\ell-mq^{d-1}\big)=mq^{d-1}-(r-1)\ell$. When $s<d-2$ and $\frac{m}{\ell}\geq\frac{1}{q^{d-s-2}}$, our repair scheme for ${\rm RS}(\mathcal{A},q^d-r)$ with $q^{s}+1\leq r\leq q^{s+1}$ given in Construction \ref{cons2} outperforms the trivial scheme because $mq^{d-1}-(r-1)\ell>mq^{d-1}-q^{s+1}\ell=q^{s+1}(mq^{d-s-2}-\ell)\geq0$. To evaluate the reduction in the I/O cost of linear repair schemes, we define the I/O cost ratio as the ratio of the I/O cost of the repair scheme to the I/O cost of the trivial repair scheme, i.e., $\rho=\frac{\gamma_{I/O}}{(n-r)\ell}$.
In \cite{I/OFormula}, some repair schemes were built for full-length RS codes, which are also better than the trivial scheme and even achieve optimal I/O cost at $r=2,3$. However, the full-length RS codes are special (i.e., $\ell=d$), and thus the I/O cost reduction is limited. In contrast, the RS codes (of the same length $n=q^d$) studied in this work are defined over a larger field $\mathbb{F}_{q^\ell}$ with $\ell-d+1\mid \ell$ or $\ell-d+2\mid \ell$. Therefore, the schemes in this work usually have lower I/O cost ratio than the schemes in \cite{I/OFormula}.
We compare the I/O cost ratio $\rho$ of the repair schemes in \cite{I/OFormula} with our scheme under the same $n,r,q=2$ but different $\ell$ \footnote{The schemes in \cite{I/OFormula} always have $\ell=\log n$, whereas the schemes from Construction \ref{cons2} in Table \ref{table4}, in turn, have $\ell=4,6,8,6,8,8$.} in Table \ref{table4}.
\begin{table}[h]
				\renewcommand\arraystretch{1.45}
				\centering
				\caption{ \scriptsize Comparison of the I/O cost ratio $\rho$ for repairing $[n,n-r]$ RS codes.}\label{table4}
				\setlength{\tabcolsep}{1mm}
				\begin{tabular}{|c|c|c|c|c|c|c|}
					\hline   $(n,r)$  & $(2^{3},2)$ & $(2^4,2)$ & $(2^5,2)$ & $(2^5,3)$ &$(2^{6},3)$ & $(2^{7},5)$ \\
					\hline  Scheme in \cite{I/OFormula}  & $94.4\%$  & $92.9\%$  & $92.7\%$ & $84.8\%$  &$85.8\%$ & $81.0\%$\\				
					\hline  Construction \ref{cons2} & $83.3\%$  & $78.6\%$  & $76.7\%$ & $79.3\%$  &$77.0\%$ & $77.2\%$  \\
					\hline
				\end{tabular}\\
				
				\vspace{6pt}
			\end{table}
			It turns out that the I/O cost of schemes in \cite{I/OFormula} achieves optimal I/O cost at $r=2$, but there is only around $7\%$ reduction in the I/O cost compared with the trivial repair scheme. For the same $(n,r=2)$, our construction achieves $23\%$ reduction in the I/O cost compared with the trivial repair scheme.

\section{Conclusion}\label{Sec6}	
In this work, we calculate the I/O cost of linear repair schemes for RS codes evaluated on subspaces and  characterize the repair bandwidth of I/O-optimal repair schemes. However, the redundancy of the RS codes under consideration remains relatively small. Extending the results to more general RS codes is a challenging work in the future.

\appendices
\section{proof of Corollary \ref{coro11}}\label{proof-coro11}
\begin{lemma}[Weil bound]\citep[Theorem 5.38]{finite field}\label{Weil}
	Let $f(x)\in F[x]$ be of degree $e\geq 1$ with $\gcd(e,{\rm Char}(F))=1$ and let $\chi$ be a nontrivial additive character of $F$. Then,
	\begin{equation*}
		\big|\sum_{\alpha\in F}\chi(f(\alpha))\big|\leq(e-1)q^{\frac{\ell}{2}}.
	\end{equation*}
\end{lemma}		
Using Theorem \ref{thmformula} and the Weil bound, we can give the proof of Corollary \ref{coro11}.
\begin{proof} Let $\mathcal{B}=\{\beta^{(1)},...,\beta^{(\ell)}\}$ be a basis of $F$ over $B$ and assume $\chi$ is the canonical additive character of $F$.
Suppose $\{g_j(x)\}_{j=1}^\ell$ is an $(m,t)$-normalized repair scheme for node $i^*$ with respect to $\mathcal{B}$. By Theorem \ref{thmformula}, the I/O cost of $\{g_j(x)\}_{j=1}^\ell$ with respect to $\mathcal{B}$ is
\begin{align*}
\gamma_{I/O}&\notag=(n-1)\ell-\frac{1}{q^m}\sum_{s=1}^t\sum_{{\bm u}\in B^m}\sum_{\alpha\in F}\chi(g_{{\bm u}}(\alpha)\beta^{(s)})\\
	      &=(n-1)\ell-\frac{1}{q^m}\Big(\sum_{s=1}^t\sum_{\alpha\in F}\chi(g_{{\bm 0}}(\alpha)\beta^{(s)})+\sum_{s=1}^t\sum_{{\bm u}\in B^m\setminus\{{\bm 0}\}}\sum_{\alpha\in F}\chi(g_{{\bm u}}(\alpha)\beta^{(s)})\Big)\\
		 &=(n-1)\ell-tq^{\ell-m}-\frac{1}{q^m}\sum_{s=1}^t\sum_{{\bm u}\in B^m\setminus\{{\bm 0}\}}\sum_{\alpha\in F}\chi(g_{{\bm u}}(\alpha)\beta^{(s)})\\
		 &\stackrel{{\rm (i)}}{\geq} (n-1)\ell-tq^{\ell-m}-\frac{1}{q^m}\sum_{s=1}^t\sum_{{\bm u}\in B^m\setminus\{{\bm 0}\}}\Big|\sum_{\alpha\in F}\chi(g_{{\bm u}}(\alpha)\beta^{(s)})\Big|\\
	 	 &\stackrel{{\rm (ii)}}{\geq}(n-1)\ell-tq^{\ell-m}-\frac{1}{q^m}\sum_{s=1}^t\sum_{{\bm u}\in B^m\setminus\{{\bm 0}\}}({\rm deg}(g_{{\bm u}})-1)q^{\frac{\ell}{2}}\\
		 &\stackrel{{\rm (iii)}}{\geq} (n-1)\ell-m\big(q^{\frac{\ell}{2}-m}+(r-2)\cdot\frac{q^m-1}{q^m}\big)q^{\frac{\ell}{2}} \\
		 &\stackrel{{\rm (iv)}}{\geq}(n-1)\ell-q^{\ell-1}-(r-2)(q-1)q^{\frac{\ell}{2}-1}.
\end{align*}
Note $\gamma_{I/O}$ and $tq^{\ell-m}$ are integers, therefore $\sum_{s=1}^t\sum_{{\bm u}\in B^m\setminus\{{\bm 0}\}}\sum_{\alpha\in F}\chi(g_{{\bm u}}(\alpha)\beta^{(s)})$ is also an integer. Combining with the triangle inequality, $({\rm i})$ follows. According to $1)$ of Definition \ref{def}, we have $0<{\rm deg}(g_{\bm u})\leq r-1\leq {\rm Char}(F)-1$ for ${\bm u}\neq{\bm 0}$, which implies that ${\rm gcd}({\rm deg}(g_{\bm u}),{\rm Char}(F))=1$, and thus the Weil bound can be applied to estimate $|\sum_{\alpha\in F}\chi(g_{{\bm u}}(\alpha)\beta^{(s)})|$, then $({\rm ii})$ holds. Finally, $({\rm iii})$ is derived from $t\leq m$ and ${\rm deg}(g_{\bm u})\leq r-1$, while $({\rm iv})$ is because $m\big(q^{\frac{\ell}{2}-m}+(r-2)\cdot\frac{q^m-1}{q^m}\big)$ reaches the maximum at $m=1$. 
\end{proof}

\section{proof of Lemma \ref{r3cond}}\label{lem17}

\begin{proof}
First note that when $t'\leq\ell-d+1$, the only constraint on $a_s$ is $0\leq a_s\leq m-1$ for $s\in[t']$, because the condition $\sum_{i=1}^{t'}a_i\leq (\ell-d+1)m$ is already implied by the previous constraint. Thus, $2^{d-m}\sum_{i=1}^{t'}2^{a_i}\leq t'2^{d-1}\leq (\ell-d+1)2^{d-1}<(\ell-d+2)2^{d-1}$. The equality can not hold in this case.

Next, we consider the case $t'\geq \ell-d+2$. For simplicity,
denote $x_i=a_i+d-m$, then it is equivalent to consider the following integer programming problem:
\begin{align}
	&\max_{t',m,x_1,...,x_{t'}} \quad   \sum_{i=1}^{t'}2^{x_i} \notag \\
	&{\rm s.t.} \left\{
	\begin{array}{ll}
d-1\geq x_1\geq x_2\geq\cdots\geq x_{t'}\geq d-m \\
(\ell-d+2)d-m\geq\sum_{i=1}^{\ell-d+2}x_i \\
\ell\geq m\geq t'\geq\ell-d+2\\
	\end{array}
	\right..\label{conds}
\end{align}
It is easy to see the maximum reaches if and only if $t'=m$ and $x_i=x_{\ell-d+2}$ for all $i\in[\ell-d+3,m]$. Denote $G(x_1,...,x_{\ell-d+1},x,m)=\sum_{i=1}^{\ell-d+1}2^{x_i}+(m-(\ell-d+1))2^{x}$. Then, the problem reduces to  	
\begin{align}
	&\max_{x_1,...,x_{\ell-d+1},x,m} \quad   G(x_1,...,x_{\ell-d+1},x,m) \notag \\
		&{\rm s.t.} \left\{
		\begin{array}{ll}
		d-1\geq x_1\geq x_2\geq\cdots\geq x_{\ell-d+1}\geq x\geq d-m \\
			(\ell-d+2)d-m-x\geq\sum_{i=1}^{\ell-d+1}x_i \\
			\ell\geq m\geq\ell-d+2\\
		\end{array}
		\right..\label{conds2}
\end{align}
Noting that the second constraint is
\begin{equation}\label{con27}
\sum_{i=1}^{\ell-d+1}x_i+x\leq(\ell-d+1)(d-1)+(\ell+1-m)
\end{equation}  and $d-1\geq\ell+1-m$,
we analyze the maximum value in the following two cases.

\begin{itemize}
  \item[(1)] If $x\leq\ell+1-m$, one can see the second constraint holds with equality at $x_1=\cdots=x_{\ell-d+1}=d-1$ and $x=\ell+1-m$, and then the maximum reaches with the value $(\ell-d+1)2^{d-1}+(m-\ell+d-1)2^{\ell+1-m}$. Denote $z=\ell+1-m$, then the latter term in the sum becomes $(d-z)2^z$. Since $1\leq z\leq d-1$, $(d-z)2^z$ reaches the maximum value $2^{d-1}$ at $z=d-1$ or $d-2$, and thus $G(x_1,...,x_{\ell-d+1},x,m)\leq (\ell-d+2)2^{d-1}$ in this case.

  \item[(2)] If $x>\ell+1-m$, i.e., $x\geq\ell+2-m$, combining (\ref{con27}) and
  the first constraint, we can deduce in this case
  $\min\{x_1,...,x_{\ell-d+1}\}\leq d-2$. As a result, it actually holds
  $\ell+2-m\leq x\leq d-2$ and thus $m\geq\ell-d+4$. Next, we analyze the influence
   of $m$ on the maximum value. For each fixed $m\in[\ell-d+4,\ell]$ and $x\in[\ell+2-m,d-2]$, denote
      $$G_m(x)=\max_{x_1,...,x_{\ell-d+1}}G(x_1,...,x_{\ell-d+1},x,m)$$
      where the maximum is taken as $x_1,...,x_{\ell-d+1}$ satisfy the constraints in (\ref{conds2}). Then we have
      the following claim.

  {\bf Claim.} $G_{m+1}(x)\leq G_{m}(x)$.

{\it proof of the claim.}
For each $x\in[\ell+2-m,d-2]$, suppose $G_{m+1}(x)=G(z_1,...,z_{\ell-d+1},x,m+1)$.
Recall that $\min\{z_1,...,z_{\ell-d+1}\}\leq d-2$. Let $i_0\in[\ell-d+1]$ be the
smallest index such that $z_{i}\leq d-2$, i.e., $z_1=\cdots=z_{i_{0}-1}=d-1, z_{i_0}\leq d-2$. For $i\in[\ell-d+1]$, set
		
$$z'_i=\begin{cases}
		z_i+1&{\rm if}~i=i_0\\
    	z_i	   &{\rm if}~ i\neq i_0\\
		\end{cases}. $$
It can be verified that $(z'_1,...,z'_{\ell-d+1},x,m)$ also
satisfies (\ref{conds2}) because $\sum_{i=1}^{\ell-d+1}z'_i=
\sum_{i=1}^{\ell-d+1}z_i+1\leq (\ell-d+2)d-m-x$, where the last inequality follows
from $\sum_{i=1}^{\ell-d+1}z_i\leq (\ell-d+2)d-(m+1)-x$ since
$(z_1,...,z_{\ell-d+1},x,m+1)$ satisfies (\ref{conds2}).		
Moreover,
\begin{align*}
	&\quad\quad G(z'_1,...,z'_{\ell-d+1},x,m)-G(z_1,...,z_{\ell-d+1},x,m+1)\\
	&=\big(\sum_{i=1}^{\ell-d+1}2^{z'_i}+(m-(\ell-d+1))2^{x}\big)-\big(\sum_{i=1}^{\ell-d+1}2^{z_i}+(m+1-(\ell-d+1))2^{x}\big)\\
	 &=2^{z_{i_0}}-2^{x}\\ &\geq0.
\end{align*}
Therefore, $G_{m+1}(x)=G(z_1,...,z_{\ell-d+1},x,m+1)\leq G(z'_1,...,z'_{\ell-d+1},x,m)\leq G_{m}(x)$. The claim is proved.

Since $m\geq\ell+2-x$ in this case, we know $G_{m}(x)\leq G_{\ell+2-x}(x)$
where
\begin{align}
	&G_{\ell+2-x}(x)=\max_{x_1,...,x_{\ell-d+1}} \quad \sum_{i=1}^{\ell-d+1}2^{x_i}+(d+1-x)2^{x} \notag \\
	&\notag{\rm s.t.} \left\{
	\begin{array}{l}
		d-1\geq x_1\geq x_2\geq\cdots\geq x_{\ell-d+1}\geq x \\
		\sum_{i=1}^{\ell-d+1}x_i\leq (\ell-d)(d-1)+d-2\\
	\end{array}
	\right..
	\end{align}	
It can be easily seen that $\sum_{i=1}^{\ell-d+1}2^{x_i}$ achieves its maximum
at $x_1=\cdots=x_{\ell-d}=d-1,x_{\ell-d+1}=d-2$. Thus,
$G_{\ell+2-x}(x)=(\ell-d)2^{d-1}+2^{d-2}+(d+1-x)2^{x}$. Since $x\leq d-2$,
it follows $(d+1-x)2^{x}$ achieves its maximum
value $3\cdot2^{d-2}$ if and only if $x=d-2$. Consequently,
\begin{equation}\label{32}
\max_{m,x}G_{m}(x)=\max_{x}G_{\ell+2-x}(x)=G_{\ell-d+4}(d-2)=(\ell-d+2)2^{d-1}.
\end{equation}
\end{itemize}
Combining the two cases, one can see $2^{d-m}\sum_{i=1}^{t'}2^{a_i}\leq
G(x_1,...,x_{\ell-d+1},x,m)\leq (\ell-d+2)2^{d-1}$. Moreover, the equality holds only if
$x_1=\cdots=x_{\ell-d}=d-1$ and $x=d-1$ or $d-2$.
Note the condition in (\ref{conds2}) imply
$(\ell-d+2)d-m-x\geq\sum_{i=1}^{\ell-d+1}x_i\geq(\ell-d+1)x$, i.e., $m\leq (\ell-d+2)(d-x)$.
Therefore, a necessary condition for the equality to hold is $m\leq2(\ell-d+2)$.		

\end{proof}

\end{sloppypar}						
\end{document}